\documentclass[11pt, final, onecolumn]{article}
\pdfoutput=1

\usepackage[backend=biber,style=alphabetic,maxbibnames=99,maxalphanames=3,natbib=true]{biblatex}
\usepackage{amsmath, amsfonts, amssymb, amsthm, graphicx, enumitem, titlesec, comment, hyperref, bbm}
\allowdisplaybreaks
\titlelabel{\thetitle.\quad}
\usepackage[margin=1in]{geometry}

\usepackage{xcolor}
\definecolor{darkblue}{rgb}{0.0,0.0,0.2}
\definecolor{darkgreen}{rgb}{0.0,0.3,0.0}
\hypersetup{colorlinks,breaklinks,
            linkcolor=darkblue,urlcolor=darkblue,
            anchorcolor=darkblue,citecolor=darkblue}

\usepackage[colorinlistoftodos,textsize=tiny]{todonotes}

\newcommand{\RR}{\mathbb{R}}

\newcommand{\PP}{\mathbb{P}}

\newcommand\abs[1]{\left\lvert#1\right\rvert}
\newcommand\ceil[1]{\left\lceil#1\right\rceil}

\newcommand\parens[1]{\left(#1\right)}
\newcommand\brackets[1]{\left[#1\right]}
\newcommand\angles[1]{\langle#1\rangle}
\newcommand*{\pr}[2][]{\PP\ifx\\\left[#1\right]\\\else_{#1}\fi \left[#2\right]}
\newcommand*{\EE}[2][]{\mathbb{E}\ifx\\\left[#1\right]\\\else_{#1}\fi \left[#2\right]}

\newcommand{\JB}{\text{JB}}

\newtheorem{defin}{Definition}[section]
\newtheorem{lemma}[defin]{Lemma}
\newtheorem{theorem}[defin]{Theorem}
\newtheorem{prop}[defin]{Proposition}
\newtheorem{corollary}[defin]{Corollary}
\newtheorem{claim}[defin]{Claim}

\theoremstyle{definition}
\newtheorem{algorithm}[defin]{Algorithm}
\newtheorem{remark}[defin]{Remark}

\definecolor{c0000ff}{RGB}{0,0,255}
\definecolor{c008000}{RGB}{0,128,0}
\definecolor{c00ffff}{RGB}{0,100,200}

\newcommand{\innerdivfig}{
\path[cm={{0.39095736,0.0,0.0,0.39095736,(-45.081848,-36.707878)}},draw=c008000,line
  join=round,line cap=rect,miter limit=4.00,line width=1.637pt]
  (131.2577,94.9154) .. controls (158.2535,149.7956) and (170.6749,180.8289) ..
  (206.3039,221.5552) -- (211.9711,227.7474) -- (217.6383,233.6854) --
  (223.3055,239.3474) -- (228.9727,244.7118) -- (234.6398,249.7576) --
  (240.3070,254.4644) -- (245.9742,258.8128) -- (251.6414,262.7841) --
  (257.3086,266.3611) -- (262.9758,269.5279) -- (268.6430,272.2703) --
  (274.3102,274.5755) -- (279.9773,276.4328) -- (282.8109,277.1908) --
  (285.6445,277.8337) -- (288.4781,278.3607) -- (291.3117,278.7713) --
  (294.1453,279.0649) -- (296.9789,279.2412) -- (299.8125,279.3000) --
  (302.6461,279.2412) -- (305.4797,279.0649) -- (308.3133,278.7713) --
  (311.1469,278.3607) -- (313.9805,277.8337) -- (316.8141,277.1908) --
  (322.4813,275.5607) -- (328.1484,273.4782) -- (333.8156,270.9530) --
  (339.4828,267.9967) -- (345.1500,264.6229) -- (350.8172,260.8467) --
  (356.4844,256.6846) -- (362.1516,252.1546) -- (367.8188,247.2758) --
  (373.4859,242.0681) -- (379.1531,236.5522) -- (384.8203,230.7495) --
  (390.4875,224.6817) -- (396.1547,218.3706) -- (404.6555,208.4959) --
  (413.1563,198.1965) -- (421.6570,187.5444) -- (479.1090,113.4900);
\path[draw=c00ffff,line join=miter,line cap=butt,miter limit=4.00,line
  width=1.6pt] (9.4912,16.0178) -- (79.8098,109.5844);
\path[draw=black,dash pattern=on 1.92pt off 1.92pt,line join=miter,line
  cap=butt,miter limit=4.00,line width=1.64pt] (78.4160,38.6101) --
  (78.4160,114.0119);
\path[shift={(-81.129466,-269.85805)},fill=black]
  (161.1647,341.3687)arc(0.000:180.000:1.500)arc(-180.000:0.000:1.500) -- cycle;
\path[shift={(-87.529466,-268.85805)},fill=black]
  (117.4655,310.4681)arc(0.000:180.000:1.500)arc(-180.000:0.000:1.500) -- cycle;
}

\newcommand{\divfig}{

\begin{tikzpicture}[y=1pt, x=1pt,yscale=-1, inner sep=0pt, outer sep=0pt]
\innerdivfig
\path[fill=black] (98.470551,91.341965) node[above right] (text6811) {$D_G(y \,\|\, x)$};
\path[fill=black] (-10,91.81398) node[above right] (text6785-9)
  {\color{c00ffff}$\hspace{-30pt}G(x)+(\;\cdot\,-x)G'(x)$};
\path[fill=black] (27.61791,33.518784) node[above right] (text6773) {$x$};
\path[fill=black] (77.02356,33.518784) node[above right] (text6777) {$y$};
\path[fill=black] (116.48392,53.68718) node[above right] (text6781) {\color{c008000}$G(\cdot)$};
\path[shift={(-81.129466,-260.85805)},fill=black]
  (161.0940,368.3094)arc(0.000:180.000:1.500)arc(-180.000:0.000:1.500) -- cycle;

\path[draw=black,miter limit=4.00,line width=1pt] (94.4643,89.5393) ..
  controls (94.4643,89.6325) and (94.3966,89.7553) .. (94.2612,89.9070) ..
  controls (94.1257,90.0588) and (93.9539,90.2222) .. (93.7455,90.3972) ..
  controls (93.5372,90.5723) and (93.3080,90.7591) .. (93.0580,90.9575) ..
  controls (92.8080,91.1560) and (92.5684,91.3486) .. (92.3393,91.5353) ..
  controls (92.1101,91.7104) and (91.9018,91.8797) .. (91.7143,92.0431) ..
  controls (91.5372,92.1948) and (91.4122,92.3232) .. (91.3393,92.4282) ..
  controls (91.2768,92.5216) and (91.2247,92.8718) .. (91.1830,93.4788) ..
  controls (91.1517,94.0857) and (91.1153,94.8269) .. (91.0737,95.7023) ..
  controls (91.0424,96.5778) and (91.0008,97.5291) .. (90.9487,98.5562) ..
  controls (90.9070,99.5717) and (90.8445,100.5405) .. (90.7612,101.4627) ..
  controls (90.6883,102.3848) and (90.5893,103.1960) .. (90.4643,103.8963) ..
  controls (90.3497,104.5967) and (90.1986,105.0636) .. (90.0112,105.2970) ..
  controls (89.8549,105.4955) and (89.5945,105.6997) .. (89.2299,105.9098) ..
  controls (88.8653,106.1199) and (88.4539,106.3183) .. (87.9955,106.5051) ..
  controls (87.5476,106.7035) and (87.0789,106.8844) .. (86.5893,107.0479) ..
  controls (86.0997,107.2230) and (85.6518,107.3689) ..
  (85.2455,107.4856)(85.2455,71.5931) .. controls (85.6518,71.7096) and
  (86.0997,71.8557) .. (86.5893,72.0308) .. controls (87.0789,72.1942) and
  (87.5476,72.3751) .. (87.9955,72.5735) .. controls (88.4539,72.7603) and
  (88.8653,72.9588) .. (89.2299,73.1688) .. controls (89.5945,73.3789) and
  (89.8549,73.5832) .. (90.0112,73.7816) .. controls (90.1986,74.0151) and
  (90.3497,74.4820) .. (90.4643,75.1823) .. controls (90.5893,75.8827) and
  (90.6882,76.6939) .. (90.7612,77.6160) .. controls (90.8445,78.5381) and
  (90.9070,79.5128) .. (90.9487,80.5399) .. controls (91.0008,81.5555) and
  (91.0424,82.5009) .. (91.0737,83.3763) .. controls (91.1154,84.2518) and
  (91.1518,84.9930) .. (91.1830,85.5999) .. controls (91.2247,86.2069) and
  (91.2767,86.5571) .. (91.3393,86.6504) .. controls (91.4434,86.8138) and
  (91.6622,87.0298) .. (91.9955,87.2982) .. controls (92.3393,87.5550) and
  (92.6934,87.8293) .. (93.0580,88.1211) .. controls (93.4226,88.4130) and
  (93.7455,88.6873) .. (94.0268,88.9440) .. controls (94.3184,89.2008) and
  (94.4643,89.3993) .. (94.4643,89.5393);
\end{tikzpicture}

}

\begin{document}
\begin{titlepage}
\title{Agreement Implies Accuracy for Substitutable Signals}

\author{Rafael Frongillo, Eric Neyman, Bo Waggoner}

\date{\today}

\maketitle
\thispagestyle{empty}

\begin{abstract}
Inspired by Aumann's agreement theorem, \citet{aar05} studied the amount of communication necessary for two Bayesian experts to approximately agree on the expectation of a random variable.
Aaronson showed that, remarkably, the number of bits does not depend on the amount of information available to each expert.
However, in general the agreed-upon estimate may be inaccurate: far from the estimate they would settle on if they were to share all of their information.
We show that if the experts' signals are \emph{substitutes}---meaning the experts' information has diminishing marginal returns---then it is the case that if the experts are close to agreement then they are close to the truth.
We prove this result for a broad class of agreement and accuracy measures that includes squared distance and KL divergence.
Additionally, we show that although these measures capture fundamentally different kinds of agreement, Aaronson's agreement result generalizes to them as well.
\end{abstract}
\end{titlepage}

\section{Introduction} \label{sec:intro}
Suppose that Alice and Bob are honest, rational Bayesians who wish to estimate some quantity---say, the unemployment rate one year from now.
Alice is an expert on historical macroeconomic trends, while Bob is an expert on contemporary monetary policy.
They convene to discuss and share their knowledge with each other until they reach an agreement about the expected value of the future unemployment rate.
Alice and Bob could reach agreement by sharing everything they had ever learned, at which point they would have the same information,
but the process would take years.
How then should they proceed?

In the seminal work ``Agreeing to Disagree," Aumann~\citep{aum76} observed that Alice and Bob can reach agreement simply by taking turns sharing their current expected value for the quantity.
In addition to modeling communication between Bayesian agents, protocols similar to this one model financial markets: each trader shares partial information about their expected value on their turn (discussed in Section~\ref{sec:markets}).
A remarkable result by Scott Aaronson \cite{aar05} shows that if Alice and Bob follow certain protocols of this form, they will agree to within $\epsilon$ with probability $1-\delta$ by communicating $O \parens{\frac{1}{\delta \epsilon^2}}$ bits.\footnote{To ensure that each message is short, Alice and Bob share discretized versions of their estimates; we discuss this in Section~\ref{sec:prelims}.}
Notably, this bound only depends on the error Alice and Bob are willing to tolerate, and not on the amount of information available to them.

Absent from Aaronson's results, however, is \emph{what Alice and Bob agree on}.
In particular, there is no guarantee that Alice and Bob will be \emph{accurate}, meaning their agreed-upon estimate will be close (in e.g.\ expected squared distance) to what they would believe if they shared all of their information.
In fact, they might agree on something highly inaccurate:
suppose that Alice and Bob have independent, uniformly random bits $b_A,b_B$, and wish to estimate the XOR $b_A \oplus b_B$.
Alice and Bob agree from the onset, as from each of their perspectives, the expected value of $b_A \oplus b_B$ is $\frac{1}{2}$.
Yet this expectation is far from the best estimate given their collective knowledge, which is either $0$ or $1$.
So while agreement is fundamental to understanding communication between Bayesians---in Aumann's terms, they cannot ``agree to disagree''---agreement is far from the whole story.
An important open problem is therefore what assumptions guarantee that Alice and Bob are accurate once they agree.

We address this open problem by introducing a natural condition, called \emph{rectangle substitutes}, under which agreement implies accuracy.
Rectangle substitutes is a notion of \emph{informational substitutes}: the property that additional information has diminishing marginal returns.
The notion of substitutes is ubiquitous in optimization problems, and informational substitutes conditions have recently been used to analyze equilibria in markets \cite{cw16}.
We show that under the rectangle substitutes condition, \emph{any} protocol leading to agreement will also lead to accuracy.
We then extend these results beyond the case of squared error, to a broad family of measures of agreement and accuracy including KL divergence.

\subsection{Overview of approach and results}

In \citet{aar05}, Alice and Bob are said to \emph{agree} if the squared distance between their estimates is small. Likewise, we can say that Alice and Bob are \emph{accurate} if the squared distance between each of their estimates and the truth is small.
In Section~\ref{sec:quadratic} we present our first main result: under these definitions, \textbf{if the information structure satisfies rectangle substitutes, then agreement implies accuracy}.
In other words, under this assumption, when two Bayesians agree---regardless of how little information they have shared---they necessarily agree \emph{on the truth}.

The proof involves carefully partitioning the space of posterior beliefs induced by the protocol.
Agreement is used to show that Alice and Bob usually fall into the same partition element, which means that Bob would not learn much from learning the partition element of Alice's expectation.
Then, the rectangle substitutes condition is used to show that if Bob were to learn Alice's partition element, then he would be very close to knowing the truth.\\

Aaronson measures agreement in terms of squared error, yet other measurements like KL divergence may be better suited for some settings.
For example, if Alice and Bob estimate the probability of a catastrophic event as $10^{-10}$ and $10^{-2}$, respectively, then under squared error they are said to agree closely, but arguably they disagree strongly, as reflected by their large KL divergence.
Motivated these different ways to measure agreement, we next ask:
\begin{enumerate}
    \item Can Aaronson's protocols be generalized to other notions of agreement, such that the number of bits communicated is independent of the amount of information available to Alice and Bob?
    \item Do other notions of agreement necessarily imply accuracy under rectangle substitutes?
\end{enumerate}
In Section \ref{sec:bregman}, we give our second and third main results: \textbf{the answer to both questions is yes.}
Specifically, the positive results apply when when measuring agreement and accuracy using Bregman divergences, a class of error measures that includes both squared distance and KL divergence.\footnote{The third result holds under an ``approximate triangle inequality" condition on the Bregman divergence, which is satisfied by most or all natural choices; indeed, it is nontrivial to construct a Bregman divergence that does not satisfy this property.}

Aaronson's proof of his agreement theorem turns out to be specific to squared distance.
Our agreement theorem (Theorem~\ref{thm:agree_bregman}) modifies Aaronson's protocol to depend on the particular Bregman divergence, i.e.\ the relevant error measure.
It then proceeds in a manner inspired by Aaronson but using several new ideas.
Our proof that agreement implies accuracy under rectangle substitutes for general Bregman divergences also involves some nontrivial changes to our proof for squared distance.
In particular, the fact that the length of an interval cannot be inferred from the Bregman divergence between its endpoints necessitates a closer analysis of the partition of Alice's and Bob's beliefs.

We conclude in Section~\ref{sec:markets} with a discussion of connections between agreement protocols and information revelation in financial markets, and discuss an interesting potential avenue for future work.

\subsection{Related Work} \label{sec:related_work}
Our setting is related to but distinct from communication complexity.
In that field (e.g.\ \cite{rao2020communication}), the goal is for Alice and Bob to correctly compute a function of their inputs while communicating as few bits as possible and using any protocol necessary.
By contrast, \citet{aar05} considered a goal of agreement, not correctness, and focused on specific natural protocols, which he showed achieve this goal in a constant number of bits.
Our work focuses on Aaronson's setting.
We discuss how our results might be framed in terms of communication complexity in Appendix~\ref{appx:comm}.

Our introduction of the substitutes condition is inspired by its usefulness in  prediction markets~\citep{cw16}.
The ``expectation-sharing'' agreement protocols we study bear a strong similarity to dynamics of market prices.
\citet{ostrovsky2012information} introduced a condition under which convergence of prices in a market implies that all information is aggregated.
This can be viewed as an ``agreement implies accuracy'' condition.
We discuss these works and the connection of our work to markets in Section~\ref{sec:markets}.
Another similar definition of informational substitutes is used by \cite{nr21b} in the context of robust aggregation of forecasts.

Finally, we note that the ``agreement protocols'' we study are not related to key agreement protocols in cryptography, where the goal is for two communicating parties to jointly  construct a shared string for cryptographic use.

\section{Preliminaries} \label{sec:prelims}
\subsection{Information Structures} We consider a set $\Omega$ of states of the world, with a probability distribution $\PP$ over the world states.
There are two experts, Alice and Bob. Alice learns the value of a random variable $\sigma: \Omega \to \mathcal{S}$; we call $\sigma$ Alice's \emph{signal} and $\mathcal{S}$ her \emph{signal set}. Correspondingly, Bob learns the value of a random variable $\tau: \Omega \to \mathcal{T}$. These signals each convey partial information about the true state $\omega \in \Omega$.
Alice and Bob are interested in a third random variable $Y: \Omega \to [0, 1]$. We use the term \emph{information structure} to refer to the tuple $\mathcal{I} := (\Omega, \PP, \mathcal{S}, \mathcal{T}, Y)$.

We denote by $\mu_{\sigma \tau} := \EE{Y \mid \sigma, \tau}$
the random variable that is equal to the expected value of $Y$ conditioned on both Alice's signal $\sigma$ and Bob's signal $\tau$.
We also define $\mu_\sigma := \EE{Y \mid \sigma}$ and $\mu_\tau := \EE{Y \mod \tau}$.
For a measurable set $S \subseteq \mathcal{S}$, we define $\mu_S := \EE{Y \mid \sigma \in S}$; we define $\mu_T$ analogously for $T \subseteq \mathcal{T}$.
Additionally, for $T \subseteq \mathcal{T}$, we define $\mu_{\sigma T} := \EE{Y \mid \tau \in T, \sigma}$, i.e.\ the expected value of $Y$ conditioned on the particular value of $\sigma$ and the knowledge that $\tau \in T$.
If Alice knows that Bob's signal belongs to $T$ (and nothing else about his signal), then the expected value of $Y$ conditional on her information is $\mu_{\sigma T}$; we refer to this as Alice's \emph{expectation}.
Likewise, for $S \subseteq \mathcal{S}$, we define $\mu_{S \tau} := \EE{Y \mid \sigma \in S, \tau}$.
Finally, we define $\mu_{ST} := \EE{Y \mid \sigma \in S, \tau \in T}$. This is the expectation of a third party who only knows that $\sigma \in S$ and $\tau \in T$.

In general we often wish to take expectations conditioned on $\sigma \in S, \tau \in T$ (for some $S \subseteq \mathcal{S}, T \subseteq \mathcal{T}$). We will use the shorthand $\EE{\cdot \mid S, T}$ for $\EE{\cdot \mid \sigma \in S, \tau \in T}$ in such cases.

\subsection{Agreement Protocols} \label{subsec:agreement-protocols}
The notion of \emph{agreement} between Alice and Bob is central to our work.
We first define agreement in terms of squared error, and generalize to other error measures in Section~\ref{sec:bregman}.

\begin{defin}[$\epsilon$-agree]
Let $a$ and $b$ be Alice's and Bob's expectations, respectively ($a$ and $b$ are random variables on $\Omega$). Alice and Bob \emph{$\epsilon$-agree} if $\frac{1}{4} \EE{(a - b)^2} \le \epsilon$.
\end{defin}
\noindent
The constant $\frac{1}{4}$ makes the left-hand side represent Alice's and Bob's distance to the average of their expectations.

Our setting follows \cite{aar05}, which examined communication protocols that cause Alice and Bob to agree.
In a \emph{(deterministic) communication protocol}, Alice and Bob take turns sending each other messages.
On Alice's turns, Alice communicates a message that is a deterministic function of her input (i.e.\ her signal $\sigma$) and all previous communication, and likewise for Bob on his turns.
A \emph{rectangle} is a set of the form $S \times T$ where $S \subseteq \mathcal{S}$ and $T \subseteq \mathcal{T}$.
The transcript of the protocol at a time step $t$ (i.e.\ after $t$ messages have been sent) partitions $\Omega$ into rectangles: for any given sequence of $t$ messages, there are subsets $S_t \subseteq \mathcal{S}, T_t \subseteq \mathcal{T}$ such that the protocol transcript at time $t$ is equal to this sequence if and only if $(\sigma, \tau) \in S_t \times T_t$.
For a given communication protocol, we may think of $S_t$ and $T_t$ as random variables.
Alice's expectation at time $t$ (i.e.\ \emph{after} the $t$-th message has been sent) is $\mu_{\sigma T_t}$ and Bob's expectation at time $t$ is $\mu_{S_t \tau}$.
Finally, the protocol terminates at a certain time (which need not be known in advance of the protocol).
While typically in communication complexity a protocol is associated with a final output, in this case we are interested in Alice's and Bob's expectations, so we do not require an output.

It will be convenient to hypothesize a third party observer, whom we call Charlie, who observes the protocol but has no other information.
At time $t$, Charlie has expectation $\mu_{S_t T_t}$.
Charlie's expectation can also be interpreted as the expectation of $Y$ according to Alice and Bob's common knowledge.

The following definition formalizes the relationship between communication protocols and agreement.

\begin{defin}[$\epsilon$-agreement protocol]
Given an information structure $\mathcal{I}$, a communication protocol \emph{causes Alice and Bob to $\epsilon$-agree} on $\mathcal{I}$ if Alice and Bob $\epsilon$-agree at the end of the protocol, i.e., if $\frac{1}{4} \EE{(a - b)^2} \le \epsilon$, where the expected value is over Alice's and Bob's inputs. We say that a communication protocol is an \emph{$\epsilon$-agreement protocol} if the protocol causes Alice and Bob to $\epsilon$-agree on every information structure.
\end{defin}

Aaronson defines and analyzes two $\epsilon$-agreement protocols.\footnote{A minor difference to our framing is that \citet{aar05} focuses on \emph{probable approximate agreement}: protocols that cause the absolute difference between Alice and Bob to be at most $\epsilon$ with probability all but $\delta$. While the results as presented in this section are stronger than those in \cite{aar05} (the original results follow from these as a consequence of Markov's inequality), these results follow from a straightforward modification of his proofs.} The first of these is the \emph{standard protocol}, in which Alice and Bob take turns stating their expectations for a number of time steps that can be computed by Alice and Bob independently in advance of the protocol, and which is guaranteed to be at most $O(1/\epsilon)$.

The fact that exchanging their expectations for $O(1/\epsilon)$ time steps results in $\epsilon$-agreement is profound and compelling. However, the standard protocol may require an unbounded number of bits of communication, since Alice and Bob are exchanging real numbers. To address this, Aaronson defines another agreement protocol that is truly polynomial-communication (which we slightly modify for our purposes):

\begin{defin}[Discretized protocol, \cite{aar05}]
Choose $\epsilon > 0$. In the \emph{discretized protocol} with parameter $\epsilon$, on her turn (at time $t$), Alice sends ``low" if her expectation is smaller than Charlie's by more than $\epsilon/4$, i.e.\ if $\mu_{S_{t - 1} \tau} < \mu_{S_{t - 1} T_{t - 1}} - \epsilon/4$; ``high" if her expectation is larger than Charlie's by more than $\epsilon/4$; and ``medium" otherwise. Bob acts analogously on his turn. At the start of the protocol, Alice and Bob use the information structure to independently compute the time $t_{\text{end}} \le \frac{1000}{\epsilon}$ that minimizes $\EE{(\mu_{\sigma T_{t_{\text{end}}}} - \mu_{S_{t_{\text{end}}} \tau})^2}$. The protocol ends at this time.
\end{defin}

\begin{theorem}[{\cite[Theorem 4]{aar05}}] \label{thm:aaronson-thm4}
The discretized protocol with parameter $\epsilon$ is an $\epsilon$-agreement protocol with transcript length $O(1/\epsilon)$ bits.
\end{theorem}

In general, we refer to Aaronson's standard and discretized protocols as examples of \emph{expectation-sharing} protocols.
We will define other examples in Section \ref{sec:bregman}, similar to Aaronson's discretized protocol but with different cutoffs for low, medium, and high.
We also interpret expectation-sharing protocols in the context of markets in Section~\ref{sec:markets}.

\subsection{Accuracy and Informational Substitutes}
Most of our main results give conditions such that if Alice and Bob $\epsilon$-agree, then Alice's and Bob's estimates are accurate.
By \emph{accurate}, we mean that Alice's and Bob's expectations are close to $\mu_{\sigma \tau}$, i.e., what they would believe if they knew each other's signals.
(After all, they cannot hope to have a better estimate of $Y$ than $\mu_{\sigma \tau}$; for this reason we sometimes refer to $\mu_{\sigma \tau}$ as the ``truth.'') Formally:

\begin{defin}[$\epsilon$-accurate]
Let $a$ be Alice's expectation. Alice is \emph{$\epsilon$-accurate} if $\EE{(\mu_{\sigma \tau} - a)^2} \le \epsilon$. We define $\epsilon$-accuracy analogously for Bob.
\end{defin}

One cannot hope for an unconditional result stating that if Alice and Bob agree, then they are accurate.
Consider for instance the \emph{XOR information structure} from the introduction: Alice and Bob each receive independent random bits as input, and $Y$ is the XOR of these bits.
Then from the start Alice and Bob agree that the expected value of $Y$ is exactly $\frac{1}{2}$, but this value is far from $\mu_{\sigma \tau}$, which is either $0$ or $1$.

Intuitively, this situation arises because Alice's and Bob's signals are \emph{informational complements}: each signal is not informative by itself, but they are informative when taken together.
On the other hand, we say that signals are \emph{informational substitutes} if learning one signal is less valuable if you already know the other signal.
An extreme example is if $\sigma=\tau=X$ for any random variable $X$.
Here $\sigma$ becomes useless upon learning $\tau$ and vice versa.
In \citet{cw16},\footnote{We recommend the ArXiv version for the most up-to-date introduction to informational substitutes.} the authors discuss formalizations of several notions of informational substitutes.
All of these notions capture ``diminishing marginal value," in the sense that, roughly speaking, the value of partial information is a submodular set function.
The various definitions proposed by \citet{cw16} only differ in how finely they allow decomposing $\sigma$ and $\tau$ to obtain a marginal unit.
Our definition has the same format, but uses a decomposition inspired by information rectangles in communication complexity.
Recall that we write $\mid S, T$ as shorthand for $\mid \sigma \in S, \tau \in T$.

\begin{defin} \label{def:rect_subs_quad}
An information structure $\mathcal{I} = (\Omega, \PP, \mathcal{S}, \mathcal{T}, Y)$ satisfies rectangle substitutes if it satisfies weak substitutes on every sub-rectangle, i.e., if for every $S \subseteq \mathcal{S}, T \subseteq \mathcal{T}$, we have
\begin{align*}
&\EE{(Y - \mu_{S \tau})^2 \mid S, T} - \EE{(Y - \mu_{\sigma \tau})^2 \mid S, T} \le \EE{(Y - \mu_{ST})^2 \mid S, T} - \EE{(Y - \mu_{\sigma T})^2 \mid S, T}.
\end{align*}
\end{defin}

We will show that under rectangle substitutes, if Alice and Bob approximately agree, then they are approximately accurate.

\paragraph{Interpreting substitutes.}
Both sides of the inequality in Definition \ref{def:rect_subs_quad} represent the ``value'' of learning $\sigma$ as measured by a decrease in error.
The left-hand side gives the decrease if one already knows $\tau$ and that $\sigma \in S$; the right-hand side gives the decrease if one only knows that $\sigma \in S, \tau \in T$.
Substitutes thus says: the marginal value of learning $\sigma$ is smaller if one already knows $\tau$ than if one does not.
This statement should hold for every sub-rectangle $S,T$.
We remark that the inequality can be rearranged to focus instead on the marginal value of $\tau$ rather than $\sigma$.
We also note that in the XOR information structure, the left-hand side of the inequality is $\frac{1}{4}$ while the right-hand side is zero: a large violation of the substitutes condition.
In the example $\sigma=\tau=X$, the left side is always zero.

\citet{cw16} discusses three interpretations of substitutes, which motivate it as a natural condition.
(1) Each side of the inequality measures an improvement in \emph{prediction error}, here the squared loss, due to learning $\sigma$.
Under substitutes, the improvement is smaller if one already knows $\tau$.
(2) Each side measures a \emph{decrease in uncertainty} due to learning $\sigma$.
Under substitutes, $\sigma$ provides less information about $Y$ if one already knows $\tau$.%
\footnote{Here, uncertainty is measured by variance of one's belief. Under the KL divergence analogue covered in Section \ref{sec:bregman_prelims}, uncertainty is measured in bits via Shannon entropy.}
(3) Each side measures the \emph{decrease in distance} of a posterior expectation from the truth when learning $\sigma$.
The distance to $Y$ changes less if one already knows $\tau$.

\subsection{The Pythagorean Theorem}
We will use the following fact throughout.
We defer the proof to Appendix~\ref{appx:bregman_omitted}, where we establish a more general version of this statement.

\begin{prop}[Pythagorean theorem] \label{prop:pythag}
Let $A$ be a random variable, $B = \EE{A \mid \mathcal{F}}$ where $\mathcal{F}$ is a sigma-algebra, and $C$ be a random variable defined on $\mathcal{F}$. Then
\[\EE{(A - C)^2} = \EE{(A - B)^2} + \EE{(B - C)^2}.\]
\end{prop}

We use the phrase \emph{Pythagorean theorem} in part because of its form, and in part because it is precisely the familiar Pythagorean theorem when the random variables are viewed as points in a Hilbert space\footnote{We do not make use of this abstraction in our work, but we refer the interested reader to \cite{zidak57}.} with inner product $\angles{X, Y} := \EE{XY}$.

Informally, $A$ is a random variable, $B$ is the expected value of $A$ conditional on some partial information, and $C$ is a random variable that only depends on this information.
So the theorem applies when $B$ is a coarse estimate of $A$ and $C$ is at least as coarse as $B$, a scenario that often occurs in our setting.

One application of the Pythagorean theorem in our context takes $A = Y$, $B = \mu_{\sigma \tau}$ (the expected value of $Y$ conditioned on the experts' signals), and $C = \mu_{\sigma T}$ (Alice's expected value, which only depends on her signal and thus on the signal pair).
This particular application, along with the symmetric one taking $C = \mu_{S \tau}$, allows us to rewrite the rectangle substitutes condition in a form that we will find more convenient:

\begin{remark}
An information structure $\mathcal{I}$ satisfies rectangle substitutes if and only if
\begin{equation} \label{eq:rec_sub_quad}
\EE{(\mu_{\sigma \tau} - \mu_{S \tau})^2 \mid S, T} \le \EE{(\mu_{\sigma T} - \mu_{ST})^2 \mid S, T}
\end{equation}
for all $S, T$.
\end{remark}

\section{Results for Squared Distance} \label{sec:quadratic}
Our main results show that, under the rectangle substitutes condition, any communication protocol that causes Alice and Bob to agree also causes them to be accurate.
We now show the first of these results, which is specific to the squared distance error measure that we have been discussing.

\subsection{Agreement Implies Accuracy}

\begin{theorem} \label{thm:agreement_accurate_quad}
Let $\mathcal{I} = (\Omega, \PP, \mathcal{S}, \mathcal{T}, Y)$ be an information structure that satisfies rectangle substitutes. For any communication protocol that causes Alice and Bob to $\epsilon$-agree on $\mathcal{I}$, Alice and Bob are $10 \epsilon^{1/3}$-accurate after the protocol terminates.
\end{theorem}

The crux of the argument is the following lemma.

\begin{lemma} \label{lem:bob_close_quad}
Let $\mathcal{I} = (\Omega, \PP, \mathcal{S}, \mathcal{T}, Y)$ be an information structure that satisfies rectangle substitutes. Let $\epsilon = \EE{(\mu_{\sigma} - \mu_{\tau})^2}$. Then
\[\EE{(\mu_{\sigma \tau} - \mu_{\tau})^2} \le 6\epsilon^{1/3}.\]
\end{lemma}

\noindent Let us first prove Theorem~\ref{thm:agreement_accurate_quad} assuming Lemma~\ref{lem:bob_close_quad} is true.

\begin{proof}[Proof of Theorem~\ref{thm:agreement_accurate_quad}]
Consider any protocol that causes Alice and Bob to $\epsilon$-agree on $\mathcal{I}$. Let $S$ be the set of possible signals of Alice at the end of the protocol which are consistent with the protocol transcript, and define $T$ likewise for Bob. Intuitively, $S \times T$ is the set of plausible signal pairs $(\sigma, \tau)$ according to an external observer of the protocol. Observe that $S$ and $T$ are random variables, each a function of both $\sigma$ and $\tau$.
We have
\begin{align*}
\EE{(\mu_{\sigma \tau} - \mu_{S \tau})^2} &= \EE[S, T]{\EE{(\mu_{\sigma \tau} - \mu_{S \tau})^2 \mid S, T}}\\
&\le \EE[S, T]{6 \parens{\EE{(\mu_{\sigma T} - \mu_{S \tau})^2 \mid S, T}}^{1/3}}\\
&\le 6 \EE[S, T]{\EE{(\mu_{\sigma T} - \mu_{S \tau})^2 \mid S, T}}^{1/3}\\
&= 6 \EE{(\mu_{\sigma T} - \mu_{S \tau})^2}^{1/3} = 6(4\epsilon)^{1/3} \le 10 \epsilon^{1/3}.
\end{align*}
In the second step, we apply Lemma~\ref{lem:bob_close_quad} to the information structure $\mathcal{I}$ restricted to $S \times T$ --- that is, to $\mathcal{I}' = (\Omega', \PP', S, T, Y)$, where $\Omega' = \{\omega \in \Omega: \sigma \in S, \tau \in T\}$ and $\PP'[\omega] = \pr{\omega \mid \sigma \in S, \tau \in T}$. (Note that we use the fact that if $\mathcal{I}$ satisfies rectangle substitutes, then so does $\mathcal{I}'$; this is because a rectangle of $\mathcal{I}'$ is also a rectangle of $\mathcal{I}$.) The third step follows by the concavity of $x^{1/3}$. Therefore, Bob is $10\epsilon^{1/3}$ accurate (and Alice is likewise by symmetry).
\end{proof}

The proof of Lemma~\ref{lem:bob_close_quad} relies on the following claim.
We defer the proof of Lemma~\ref{lem:bob_close_quad} (and Claim~\ref{claim:n_sigma_tau}) to Appendix~\ref{appx:quad_omitted}, and instead sketch the proofs here.

\begin{claim} \label{claim:n_sigma_tau}
For any $N \ge 1$, it is possible to partition $[0, 1]$ into $N$ intervals $[0, x_1), [x_1, x_2), \dots,$ $[x_{N - 1}, 1]$ in a way so that each interval has length at most $\frac{2}{N}$, and
\[\pr{k(\sigma) \neq k(\tau)} \le \sqrt{\epsilon} N,\]
where $k(\sigma)$ denotes the $k \in [N]$ such that $x_{k - 1} \le \mu_{\sigma} < x_k$, and $k(\tau)$ is defined analogously.\footnote{For convenience we define $x_0 = 0$ and $x_N$ to be some number greater than $1$.}
\end{claim}
Intuitively, Claim~\ref{claim:n_sigma_tau} is true because if $\EE{(\mu_{\sigma} - \mu_{\tau})^2}$ is small, then $\mu_{\sigma}$ and $\mu_{\tau}$ are likely to fall into the same interval.

We now sketch the proof of Lemma~\ref{lem:bob_close_quad}.
To see why Claim~\ref{claim:n_sigma_tau} is relevant, recall that we wish to upper bound the expectation of $(\mu_{\sigma \tau} - \mu_{\tau})^2$. Let $S^{(k)} := \{\sigma \in \mathcal{S}: x_{k - 1} \le \mu_{\sigma} < x_k\}$. By the Pythagorean theorem, we have
\[\EE{(\mu_{\sigma \tau} - \mu_{\tau})^2} = \EE{(\mu_{\sigma \tau} - \mu_{S^{(k(\sigma))} \tau})^2} + \EE{(\mu_{S^{(k(\sigma))} \tau} - \mu_{\tau})^2}.\]
By using the rectangle substitutes condition for $S = S^{(k)}, T = \mathcal{T}$ for every $k$, we find that
\begin{equation} \label{eq:rect_subs_app}
\EE{(\mu_{\sigma} - \mu_{S^{(k(\sigma))}})^2} \ge \EE{(\mu_{\sigma \tau} - \mu_{S^{(k(\sigma))} \tau})^2}.
\end{equation}
Therefore, we have
\begin{equation} \label{eq:two_sums_quad}
\EE{(\mu_{\sigma \tau} - \mu_{\tau})^2} \le \EE{(\mu_{\sigma} - \mu_{S^{(k(\sigma))}})^2} + \EE{(\mu_{S^{(k(\sigma))} \tau} - \mu_{\tau})^2}.
\end{equation}
Claim~\ref{claim:n_sigma_tau} lets us argue that the first of these two terms is small (because $\mu_{\sigma}$ and $\mu_{S^{(k(\sigma))}}$ are always within $\frac{2}{N}$ of each other) and that the second term is also small (because conditioned on $\tau$, $k(\sigma)$ is known with high probability). We find that choosing $N = \epsilon^{1/6}$ gives us the bound in Lemma~\ref{lem:bob_close_quad}.

Theorem~\ref{thm:agreement_accurate_quad} is a general result about agreement protocols. Applying the result to Aaronson's discretized protocol gives us the following result.

\begin{corollary} \label{cor:aaronson_exp}
Let $\mathcal{I}$ be any information structure that satisfies universal rectangle substitutes. For any $\epsilon > 0$, Alice and Bob will be $\epsilon$-accurate after running Aaronson's discretized protocol with parameter $\epsilon^3/1000$ (and this takes $O(1/\epsilon^3)$ bits of communication).
\end{corollary}

\begin{remark}
The discretized protocol is not always the most efficient agreement protocol.
For example, Proposition~\ref{prop:fast_rect} shows that if the rectangle substitutes condition holds, agreement (and therefore accuracy) can be reached with just $O(\log(1/\epsilon))$ bits, an improvement on Corollary \ref{cor:aaronson_exp}.
We discuss communication complexity further in Appendix~\ref{appx:comm}.
Even if more efficient protocols are sometimes possible, expectation-sharing protocols are of interest because they model naturally-occurring communication processes.
For example, they capture the dynamics of prices in markets, which we also discuss in Section~\ref{sec:markets}.
More generally, we find it remarkable that Alice and Bob become accurate by running the agreement protocol (indeed \emph{any} agreement protocol), despite such protocols being designed with only agreement in mind.
\end{remark}

Finally, we observe the following important consequence of Theorem~\ref{thm:agreement_accurate_quad}: once Alice and Bob agree, they continue to agree.

\begin{corollary}
\label{cor:continue-agreeing-quad}
Let $\mathcal{I} = (\Omega, \PP, \mathcal{S}, \mathcal{T}, Y)$ be an information structure that satisfies rectangle substitutes. Consider a communication protocol with the property that Alice and Bob $\epsilon$-agree after round $t$. Then Alice and Bob $10 \epsilon^{1/3}$-agree on all subsequent time steps.
\end{corollary}

\begin{proof}
If Alice and Bob $\epsilon$-agree then they are $10\epsilon^{1/3}$-accurate, so in particular $\EE{(\mu_{\sigma \tau} - \mu_{\sigma T_t})^2} \le 10\epsilon^{1/3}$. Note that $\EE{(\mu_{\sigma \tau} - \mu_{\sigma T_s})^2}$ is a decreasing function of $s$, since for any $s_1 \le s_2$ we have
\[\EE{(\mu_{\sigma \tau} - \mu_{\sigma T_{s_1}})^2} = \EE{(\mu_{\sigma \tau} - \mu_{\sigma T_{s_2}})^2} + \EE{(\mu_{\sigma T_{s_2}} - \mu_{\sigma T_{s_1}})^2}\]
by the Pythagorean theorem. Therefore, for any $t' > t$, we have $\EE{(\mu_{\sigma \tau} - \mu_{\sigma T_{t'}})^2} \le 10\epsilon^{1/3}$. Symmetrically, we have $\EE{(\mu_{\sigma \tau} - \mu_{S_{t'} \tau})^2} \le 10\epsilon^{1/3}$. Therefore, $\EE{(\mu_{\sigma T_{t'}} - \mu_{S_{t'} \tau})^2} \le 40\epsilon^{1/3}$, which means that after round $t'$, Alice and Bob $10 \epsilon^{1/3}$-agree.
\end{proof}

Corollary~\ref{cor:continue-agreeing-quad} stands in contrast to the more general case, in which it is possible that Alice and Bob ``nearly agree for the first $t - 1$ time steps, then disagree violently at the $t$-th step" \cite[\S2.2]{aar05}.
Thus, while the main purpose of Theorem~\ref{thm:agreement_accurate_quad} is a property about \emph{accuracy}, an \emph{agreement} property falls out naturally: under the rectangle substitutes condition, once Alice and Bob are close to agreement, they will remain in relatively close agreement into the future.

\subsection{Graceful Decay Under Closeness to Rectangle Substitutes}
In a sense, the rectangle substitutes condition is quite strong: it requires that the weak substitutes condition be satisfied on \emph{every} sub-rectangle. One might hope for a result that generalizes Theorem~\ref{thm:agreement_accurate_quad} to information structures that almost satisfy the rectangle substitutes condition but do not quite.
Let us formally define a notion of closeness to rectangle substitutes.

\begin{defin}
An information structure $\mathcal{I} = (\Omega, \PP, \mathcal{S}, \mathcal{T}, Y)$ satisfies \emph{$\delta$-approximate rectangle substitutes} if for every partition of $\mathcal{S} \times \mathcal{T}$ into rectangles,\footnote{There are partitions into rectangles that cannot arise from a communication protocol. Our results would apply equally if this condition were instead defined for every partition that could arise from a communication protocol, but we state this condition more generally so that it could be applicable in a broader context than the analysis of communication protocols.} the rectangle substitutes condition holds in expectation over the partition, up to an additive constant of $\delta$, i.e., if we have
\begin{equation} \label{eq:approx_rect}
\EE[\sigma, \tau]{(\mu_{\sigma \tau} - \mu_{S_{\sigma, \tau} \tau})^2} \le \EE[\sigma, \tau]{(\mu_{\sigma T_{\sigma, \tau}} - \mu_{S_{\sigma, \tau} T_{\sigma, \tau}})^2} + \delta,
\end{equation}
where $S_{\sigma, \tau} \times T_{\sigma, \tau}$ is the rectangle containing $(\sigma, \tau)$.
\end{defin}

\begin{remark}
The $\delta$-approximate rectangle substitutes property is a relaxation of the rectangle substitutes property, in the sense that the two are equivalent if $\delta = 0$. To see this, first observe that if $\mathcal{I}$ satisfies rectangle substitutes, then it satisfies Equation~\ref{eq:approx_rect} with $\delta = 0$ pointwise across all $S_{\sigma, \tau}, T_{\sigma, \tau}$, and thus in expectation. In the other direction, suppose that $\mathcal{I}$ satisfies $0$-approximate rectangle substitutes. Let $S \subseteq \mathcal{S}, T \subseteq \mathcal{T}$ and consider the partition of $\mathcal{I}$ into rectangles that contains $S \times T$ and, separately, every other signal pair $(\sigma, \tau)$ in its own rectangle. For this partition, Equation~\ref{eq:approx_rect} reduces precisely to Equation~\ref{eq:rec_sub_quad} (the rectangle substitutes condition for $S$ and $T$).
\end{remark}

\noindent Theorem~\ref{thm:agreement_accurate_quad} generalizes to approximate rectangle substitutes as follows.

\begin{theorem} \label{thm:agreement_accurate_graceful_decay}
Let $\mathcal{I} = (\Omega, \PP, \mathcal{S}, \mathcal{T}, Y)$ be an information structure that satisfies $\delta$-approximate rectangle substitutes. For any communication protocol that causes Alice and Bob to $\epsilon$-agree on $\mathcal{I}$, Alice and Bob are $(10 \epsilon^{1/3} + \delta)$-accurate after the protocol terminates.
\end{theorem}

\begin{proof}
We first observe that Lemma~\ref{lem:bob_close_quad} can be modified as follows.

\begin{lemma} \label{lem:bob_close_graceful_decay}
Let $\mathcal{I} = (\Omega, \PP, \mathcal{S}, \mathcal{T}, Y)$ be an information structure that satisfies $\delta$-approximate rectangle substitutes. Let $\epsilon = \EE{(\mu_{\sigma} - \mu_{\tau})^2}$. Then
\[\EE{(\mu_{\sigma \tau} - \mu_{\tau})^2} \le 6\epsilon^{1/3} + \delta.\]
\end{lemma}

\noindent
The proof of Lemma~\ref{lem:bob_close_graceful_decay} is exactly the same as that of Lemma~\ref{lem:bob_close_quad}, except that Equation~\ref{eq:rect_subs_app} (Equation~\ref{eq:delta_change} in the full proof) includes an additive $\delta$ term on the left-hand side:
\[\EE{(\mu_{\sigma} - \mu_{S^{(k(\sigma))}})^2} + \delta \ge \EE{(\mu_{\sigma \tau} - \mu_{S^{(k(\sigma))} \tau})^2}.\]
This modified inequality follows immediately from the $\delta$-approximate rectangle substitutes condition, noting that one partition of $\mathcal{S} \times \mathcal{T}$ into rectangles is $\{S_1 \times \mathcal{T}, \dots, S_N \times \mathcal{T}\}$.
The extra $\delta$ term produces the $\delta$ term in the lemma statement.

To prove the theorem, let $S$ be the set of possible signals of Alice at the end of the protocol which are consistent with the protocol transcript, and define $T$ likewise for Bob. Let $\delta_{ST}$ be the minimum $\delta$ such that $S \times T$ satisfies $\delta$-approximate rectangle substitutes. Note that $\EE[S, T]{\delta_{ST}} \le \delta$: otherwise, by taking the union over the worst-case partitions for each $S, T$ we would exhibit a partition of $\mathcal{S} \times \mathcal{T}$ into rectangles that would violate the $\delta$-approximate rectangle substitutes property. Therefore we have
\begin{align*}
\EE{(\mu_{\sigma \tau} - \mu_{S \tau})^2} &= \EE[S, T]{\EE{(\mu_{\sigma \tau} - \mu_{S \tau})^2 \mid S, T}}\\
&\le \EE[S, T]{6 \parens{\EE{(\mu_{\sigma T} - \mu_{S \tau})^2 \mid S, T}}^{1/3} + \delta_{ST}}\\
&\le 6 \EE[S, T]{\EE{(\mu_{\sigma T} - \mu_{S \tau})^2 \mid S, T}}^{1/3} + \delta\\
&= 6 \EE{(\mu_{\sigma T} - \mu_{S \tau})^2}^{1/3} + \delta = 6(4\epsilon)^{1/3} + \delta \le 10 \epsilon^{1/3} + \delta.
\end{align*}
As in the proof of Theorem~\ref{thm:agreement_accurate_quad}, the second step follows by applying Lemma~\ref{lem:bob_close_quad} to the information structure $\mathcal{I}$ restricted to $S \times T$.
\end{proof}

\section{Results for Other Divergence Measures} \label{sec:bregman}
Squared distance is a compelling error measure because it elicits the mean. That is, if you wish to estimate a random variable $Y$ and will be penalized according to the squared distance between $Y$ and your estimate, the strategy that minimizes your expected penalty is to report the expected value of $Y$ (conditional on the information you have). This is in contrast to e.g.\ absolute distance as an error measure, which would instead elicit the median of your distribution. The class of error measures that elicit the mean is precisely the class of \emph{Bregman divergences} (defined below).

In this section, our main result is a \textbf{generalization of Theorem~\ref{thm:agreement_accurate_quad} to (almost) arbitrary Bregman divergences} (see e.g.\ Theorem~\ref{thm:agreement_accurate_beta}). Additionally, we provide a \textbf{generalization of Aaronson's discretized protocol to arbitrary Bregman divergences} (Theorem~\ref{thm:agree_bregman}).

\subsection{Preliminaries on Bregman Divergences} \label{sec:bregman_prelims}
\begin{defin}
Given a differentiable,\footnote{When we say ``differentiable," we mean differentiable on the interior of the interval on which $G$ is defined.} strictly convex function $G: [0, 1] \to \RR$, and $x, y \in [0, 1]$, the \emph{Bregman divergence} from $y$ to $x$ is
\[D_G(y \parallel x) := G(y) - G(x) - (y - x)G'(x).\]
\end{defin}

\begin{figure}
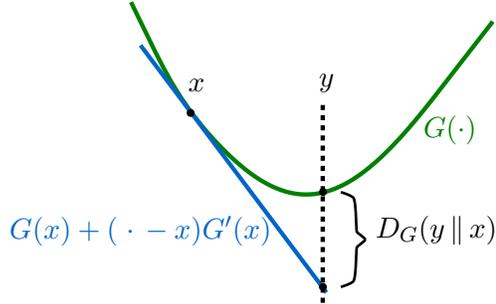

    \centering
    \divfig
    \caption{The Bregman divergence $D_G(y \parallel x)$ is the vertical distance at $y$ between $G$ and the tangent line to $G$ at $x$.}
    \label{fig:bregman}
\end{figure}

\begin{prop}[\cite{banerjee2005clustering}] \label{prop:bregman_exp}
Given a random variable $Y$, the quantity $\EE{D_G(Y \parallel x)}$ is minimized by $x = \EE{Y}$.
\end{prop}

An intuitive formulation of Bregman divergence is that $D_G(y \parallel x)$ can be found by drawing the line tangent to $G$ at $x$ and computing how far below the point $(y, G(y))$ this line passes. We illustrate this in Figure~\ref{fig:bregman}. Note that the Bregman divergence is not in general symmetric in its arguments; indeed, $G(x) = x^2$ is the only $G$ for which it is.

The Bregman divergence with respect to $G(x) = x^2$ is precisely the squared distance. Another common Bregman divergence is the KL divergence, which corresponds to $G(x) = x \log x + (1 - x) \log(1 - x)$, the negative of Shannon entropy.

We generalize relevant notions such as agreement and accuracy to arbitrary Bregman divergences as follows. In the definitions below, $G: [0, 1] \to \RR$ is a differentiable, strictly convex function.

\begin{defin}
Let $a$ be Alice's expectation. Alice is \emph{$\epsilon$-accurate} if $\EE{D_G(\mu_{\sigma \tau} \parallel a)} \le \epsilon$, and likewise for Bob.
\end{defin}

We discuss our choice of the order of these two arguments (i.e.\ why we do not instead consider the expectation of $D_G(a \parallel \mu_{\sigma \tau})$) in Appendix~\ref{appx:alternative_defs}. We now define $\epsilon$-agreement, and to do so we first define the \emph{Jensen-Bregman divergence}.

\begin{defin}
For $a, b \in [0, 1]$, the \emph{Jensen-Bregman divergence} between $a$ and $b$ with respect to $G$ is
\[\JB_G(a, b) := \frac{1}{2} \parens{D_G \parens{a \parallel \frac{a + b}{2}} + D_G \parens{b \parallel \frac{a + b}{2}}} = \frac{G(a) + G(b)}{2} - G \parens{\frac{a + b}{2}}.\]
\end{defin}

The validity of the second equality can be easily derived from the definition of Bregman divergence. Note that the Jensen-Bregman divergence, unlike the Bregman divergence, is symmetric in its arguments. The Jensen-Bregman divergence is a lower bound on the average Bregman divergence from Alice and Bob to any other point (see Proposition~\ref{prop:bregman_facts}~\ref{item:min_bregman}).

\begin{defin}
Let $a$ and $b$ be Alice's and Bob's expectations, respectively. Alice and Bob \emph{$\epsilon$-agree} with respect to $G$ if $\JB_G(a, b) \le \epsilon$.
\end{defin}

In Appendix~\ref{appx:alternative_defs} we discuss alternative definitions of agreement and accuracy. The upshot of this discussion is that our definition of agreement is the \emph{weakest} reasonable one, and our definition of accuracy is the \emph{strongest} reasonable one. This means that the main result of this section---that under a wide class of Bregman divergence, agreement implies accuracy---is quite powerful: it starts with a weak premise and proves a strong conclusion.

\begin{defin}
Given an information structure $\mathcal{I}$, a communication protocol \emph{causes Alice and Bob to $\epsilon$-agree} on $\mathcal{I}$ with respect to $G$ if Alice and Bob $\epsilon$-agree with respect to $G$ at the end of the protocol. A communication protocol is an \emph{$\epsilon$-agreement protocol} with respect to $G$ if the protocol causes Alice and Bob to $\epsilon$-agree with respect to $G$ on every information structure.
\end{defin}

We also generalize the notion of rectangle substitutes to this domain, following \cite{cw16}, which explored notions of substitutes for arbitrary Bregman divergences.

\begin{defin} \label{def:rect_subs_bregman}
Let $G: [0, 1] \to \RR$ be a differentiable, strictly convex function. An information structure $\mathcal{I} = (\Omega, \PP, \mathcal{S}, \mathcal{T}, Y)$ satisfies rectangle substitutes with respect to $G$ if for every $S \subseteq \mathcal{S}, T \subseteq \mathcal{T}$, we have
\begin{align*}
&\EE{D_G(Y \parallel \mu_{S \tau}) \mid S, T} - \EE{D_G(Y \parallel \mu_{\sigma \tau}) \mid S, T}\\
&\le \EE{D_G(Y \parallel \mu_{ST}) \mid S, T} - \EE{D_G(Y \parallel \mu_{\sigma T}) \mid S, T}.
\end{align*}
\end{defin}

The Pythagorean theorem (Proposition~\ref{prop:pythag}) generalizes to arbitrary Bregman divergences:
\begin{prop} \label{prop:pythag_bregman}
Let $A$ be a random variable, $B = \EE{A \mid \mathcal{F}}$ where $\mathcal{F}$ is a sigma-algebra, and $C$ be a random variable defined on $\mathcal{F}$. Then
\[\EE{D_G(A \parallel C)} = \EE{D_G(A \parallel B)} + \EE{D_G(B \parallel C)}.\]
\end{prop}
Although the proof of this observation is fairly straightforward, to our knowledge Proposition~\ref{prop:pythag_bregman} is original to this work. 
We provide a proof in Appendix~\ref{appx:bregman_omitted}.
Just as we did with squared error, this general Pythagorean theorem allows us to rewrite the rectangle substitutes condition for Bregman divergences.

\begin{remark}
An information structure $\mathcal{I}$ satisfies rectangle substitutes with respect to $G$ if and only if for all $S \subseteq \mathcal{S}, T \subseteq \mathcal{T}$ we have
\begin{equation} \label{eq:rec_sub_bregman}
\EE{D_G(\mu_{\sigma \tau} \parallel \mu_{S \tau}) \mid S, T} \le \EE{D_G(\mu_{\sigma T} \parallel \mu_{ST}) \mid S, T}.
\end{equation}
Given the interpretation of Bregman divergences as measures of error, we can interpret the left side as Bob's expected error in predicting the truth while the right side is Charlie's expected error when predicting Alice's expectation.
Both sides measure a prediction error due to not having Alice's signal, but from different starting points.
\end{remark}

\subsection{Generalizing the Discretized Protocol}
Later in this work, we will show that under some weak conditions, protocols that cause Alice and Bob to agree with respect to $G$ also cause Alice and Bob to be accurate with respect to $G$. However, this raises an interesting question: are there protocols that cause Alice and Bob to agree with respect to $G$?
In particular, we are interested in natural expectation-sharing protocols.
Aaronson's discretized protocol is specific to $G(x) = x^2$, and it is not immediately obvious how to generalize it. We present the following generalization.

\begin{defin}
Let $G$ be a differentiable, strictly convex function, and let $M := \max_x G(x) - \min_x G(x)$. Choose $\epsilon > 0$. In the \emph{discretized protocol with respect to $G$} with parameter $\epsilon$, on her turn (at time $t$), Alice sends ``medium" if $D_G(\mu_{\sigma T_{t - 1}} \parallel \mu_{S_{t - 1} T_{t - 1}}) < \frac{\epsilon}{2}$, and otherwise either ``low" or ``high", depending on whether $\mu_{\sigma T_t}$ is smaller or larger (respectively) than $\mu_{S_t T_t}$. Bob acts analogously on his turn. At the start of the protocol, Alice and Bob use the information structure to independently compute the time $t_{\text{end}} \le \frac{24M(4M + \epsilon)}{\epsilon^2}$ that minimizes $\EE{D_G(\mu_{\sigma T_{t_{\text{end}}}} \parallel \mu_{S_{t_{\text{end}}} \tau})}$. The protocol ends at this time.
\end{defin}

\begin{theorem} \label{thm:agree_bregman}
The discretized protocol with respect to $G$ with parameter $\epsilon$ is an $\epsilon$-agreement protocol that involves $O \parens{\frac{M(M + \epsilon)}{\epsilon^2}}$ bits of communication.
\end{theorem}

Our proof draws inspiration from Aaronson's proof of the discretized protocol, but has significant differences. The key idea is to keep track of the monovariant $\EE{D_G(Y \parallel \mu_{S_t T_t})}$. This is Charlie's expected error (as measured by the Bregman divergence from the correct answer $Y$) after time step $t$---recall that Charlie is our name for a third-party observer of the protocol. Note that this quantity is at most $M$ and at least $0$. Hence, if we show that the quantity decreases by at least some value $\beta$ every time Alice and Bob do not $\epsilon$-agree, then we will have shown that Alice and Bob must $\epsilon$-agree within $\frac{\beta}{M}$ time steps. We defer the proof to Appendix~\ref{appx:bregman_omitted}.

\subsection{Approximate Triangle Inequality}

Our results will hold for a class of Jensen-Bregman divergences that satisfy an approximate version of the triangle inequality.
Specifically, we will require $\JB_G$ to satisfy the following \emph{$c$-approximate triangle inequality} for some $c>0$.

\begin{defin}
Given a differentiable, strictly convex function $G: [0, 1] \to \RR$ and a positive number $c$, we say that $\JB_G(\cdot, \cdot)$ satisfies the \emph{$c$-approximate triangle inequality} if for all $a, b, x \in [0, 1]$ we have
\[\JB_G(a, x) + \JB_G(x, b) \ge c \JB_G(a, b).\]
\end{defin}

It is possible to construct functions $G$ such that there is no positive $c$ for which $\JB_G$ satisfies the $c$-approximate triangle inequality. However, $\JB_G$ satisfies the $c$-approximate triangle inequality for some positive $c$ for essentially all natural choices of $G$.

\begin{prop} \label{prop:triangle} Let $G: [0, 1] \to \RR$ be a differentiable, strictly convex function.
\begin{enumerate}[label=(\roman*)]
\item \label{item:capprox_1} If $\sqrt{\JB_G(\cdot, \cdot)}$ satisfies the triangle inequality, then $\JB_G$ satisfies the $\frac{1}{2}$-approximate triangle inequality.
\item \label{item:capprox_2} If $G(x) = x^2$ (i.e.\ $D_G$ is squared distance) or if $G(x) = x \log x + (1 - x) \log(1 - x)$ (i.e.\ $D_G$ is KL divergence), then $\sqrt{\JB_G}$ satisfies the triangle inequality (and so $\JB_G$ satisfies the $\frac{1}{2}$-approximate triangle inequality).
\end{enumerate}
\end{prop}

\begin{proof}
Regarding Fact~\ref{item:capprox_1}, suppose that $\sqrt{\JB_G}$ satisfies the triangle inequality. Then for all $a, b, x$ we have $\sqrt{\JB_G(a, x)} + \sqrt{\JB_G(x, b)} \ge \sqrt{\JB_G(a, b)}$. Squaring both sides and observing that $\JB_G(a, x) + \JB_G(x, b) \ge 2\sqrt{\JB_G(a, x) \JB_G(x, b)}$ completes the proof.

Fact~\ref{item:capprox_2} is trivial for $G(x) = x^2$, since $\sqrt{\JB_G}$ is the absolute distance metric (times a constant factor). As for $G(x) = x \log x + (1 - x) \log(1 - x)$, we refer the reader to \cite{es03}.
\end{proof}

The question of when $\sqrt{\JB_G}$ satisfies the triangle inequality has been explored in previous work; we refer the interested reader to \cite{abb13} and \cite{ccr08}.\\

\subsection{Generalized Agreement Implies Generalized Accuracy}
In all of the results in this subsection, we consider the following setting: $G$ is a differentiable convex function; $c$ is a positive real number such that $\JB_G$ satisfies the $c$-approximate triangle inequality; and $\mathcal{I} = (\Omega, \PP, \mathcal{S}, \mathcal{T}, Y)$ is an information structure that satisfies rectangle substitutes with respect to $G$.

We prove generalizations of Theorem~\ref{thm:agreement_accurate_quad}, showing that under the rectangle substitutes condition, if a protocol ends with Alice and Bob in approximate agreement, then Alice and Bob are approximately accurate.
The first result we state assumes that $G$ is symmetric, but is otherwise quite general.

\begin{theorem} \label{thm:agreement_accurate_beta}
Assume that $G$ is symmetric about the line $x = \frac{1}{2}$. For any communication protocol that causes Alice and Bob to $\epsilon$-agree on $\mathcal{I}$, and for any $\beta \ge \frac{2}{c} \epsilon$, Alice and Bob are
\[\parens{\frac{8}{c^2} \beta + 16 \parens{G(0) - G\parens{\parens{\frac{\epsilon}{\beta}}^{1/(1 - \log_2 c)}}}}\text{-accurate}\]
after the protocol terminates.
\end{theorem}

This result is not our most general, as it assumes that $G$ is symmetric, but this assumption likely holds for most use cases. To apply the result optimally, one must first optimize $\beta$ as a function of $G$. For example, setting $\beta = \epsilon^{r/(r + 1 - \log_2 c)}$ (with $r$ defined below) gives us the following corollary:\footnote{Corollary~\ref{cor:agreement_accurate_bregman_r} as stated (without the symmetry assumption) is actually a corollary of Theorem~\ref{thm:agreement_accurate_bregman}.}

\begin{corollary} \label{cor:agreement_accurate_bregman_r}
Assume that $G(0) - G(x), G(1) - G(1 - x) \le O(x^r)$. For any communication protocol that causes Alice and Bob to $\epsilon$-agree on $\mathcal{I}$, Alice and Bob are $O \parens{\epsilon^{r/(r + 1 - \log_2 c)}}$-accurate after the protocol terminates, where the constant hidden by $O(\cdot)$ depends on $G$.
\end{corollary}

\begin{remark}
Concretely, if $G'$ is bounded then we can choose $r = 1$, in which case our bound simplifies to $O \parens{\epsilon^{1/(2 - \log_2 c)}}$. If instead we assume that $c = \frac{1}{2}$ (as is the case if $\sqrt{\JB_G(\cdot, \cdot)}$ is a metric), then the bound is $O \parens{\epsilon^{r/(r + 2)}}$. If both of these are true, as is the case for $G(x) = x^2$, then the bound is $O(\epsilon^{1/3})$, which recovers our result in Theorem~\ref{thm:agreement_accurate_quad}.
\end{remark}

For $G$ equal to the negative of Shannon entropy (i.e.\ the $G$ for which $D_G$ is KL divergence), setting $\beta = \epsilon^{1/3} (\log 1/\epsilon)^{2/3}$ in Theorem~\ref{thm:agreement_accurate_beta} gives us the following corollary.

\begin{corollary} \label{cor:agreement_accurate_bregman_log}
If $G(x) = x \log x + (1 - x) \log(1 - x)$, then for any communication protocol that causes Alice and Bob to $\epsilon$-agree on $\mathcal{I}$, Alice and Bob are $O(\epsilon^{1/3} (\log 1/\epsilon)^{2/3})$-accurate after the protocol terminates.
\end{corollary}

Theorem~\ref{thm:agreement_accurate_beta} follows from our most general result about agreement implying accuracy:

\begin{theorem} \label{thm:agreement_accurate_bregman}
Let $\tilde{G}(x) := \max_{a, b: \abs{a - b} \le x}(G(a) - G(b))$ be the maximum possible difference in $G$-values of two points that differ by at most $x$, and let $\tilde{G}^*(x)$ be the concave envelope of $\tilde{G}$, i.e.
\[\tilde{G}^*(x) := \max_{0 \le a, b, w \le 1: wa + (1 - w)b = x} w \tilde{G}(a) + (1 - w) \tilde{G}(b).\]
For any communication protocol that causes Alice and Bob to $\epsilon$-agree on $\mathcal{I}$, and for any $\beta > 0$, Alice and Bob are
\[\parens{\frac{8}{c^2} \beta + 16 \tilde{G}^* \parens{\parens{\frac{\epsilon}{\beta}}^{1/(1 - \log_2 c)}}} \text{-accurate}\]
after the protocol terminates.
\end{theorem}

\begin{proof}
To prove Theorem~\ref{thm:agreement_accurate_bregman}, it suffices to prove the following lemma.
\begin{lemma} \label{lem:bob_close_bregman}
Let $G$ be a differentiable convex function on $[0, 1]$ and $c \in (0, 1)$ be such that $\JB_G$ satisfies the $c$-approximate triangle inequality. Let $\mathcal{I} = (\Omega, \PP, \mathcal{S}, \mathcal{T}, Y)$ be an information structure that satisfies rectangle substitutes with respect to $G$. Let $\epsilon = \EE{\JB_G(\mu_{\sigma} \parallel \mu_{\tau})}$. Then for any $\beta > 0$, we have
\[\EE{D_G(\mu_{\sigma \tau} \parallel \mu_{\tau})} \le \frac{8}{c^2} \beta + 16 \tilde{G}^* \parens{\parens{\frac{\epsilon}{\beta}}^{1/(1 - \log_2 c)}}.\]
\end{lemma}

\noindent Let us first prove Theorem~\ref{thm:agreement_accurate_bregman} assuming Lemma~\ref{lem:bob_close_bregman} is true.\\

Consider any protocol that causes Alice and Bob to $\epsilon$-agree on $\mathcal{I}$. Let $S$ be the set of possible signals of Alice at the end of the protocol which are consistent with the protocol transcript, and define $T$ likewise for Bob.

Let $\epsilon_{ST} = \EE{\JB_G(\mu_{\sigma T}, \mu_{S \tau}) \mid S, T}$. Note that
\[\EE[S, T]{\epsilon_{ST}} = \EE[S, T]{\EE{\JB_G(\mu_{\sigma T}, \mu_{S \tau}) \mid S, T}} = \EE{\JB_G(\mu_{\sigma T}, \mu_{S \tau})} \le \epsilon.\]
Therefore, for any $\beta > 0$ we have
\begin{align*}
\EE{D_G(\mu_{\sigma \tau} \parallel \mu_{S \tau})} &\le \EE[S, T]{\frac{8}{c^2} \beta + 16 \tilde{G}^* \parens{\parens{\frac{\epsilon_{ST}}{\beta}}^{1/(1 - \log_2 c)}}}\\
&\le \frac{8}{c^2} \beta + 16 \tilde{G}^* \parens{\EE[S, T]{\parens{\frac{\epsilon_{ST}}{\beta}}^{1/(1 - \log_2 c)}}}\\
&\le \frac{8}{c^2} \beta + 16 \tilde{G}^* \parens{{\parens{\frac{\EE[S, T]{\epsilon_{ST}}}{\beta}}^{1/(1 - \log_2 c)}}}\\
&\le \frac{8}{c^2} \beta + 16 \tilde{G}^* \parens{\parens{\frac{\epsilon}{\beta}}^{1/(1 - \log_2 c)}}.
\end{align*}
In the first step, we apply Lemma~\ref{lem:bob_close_bregman} to the information structure $\mathcal{I}$ restricted to $S \times T$---that is, to $\mathcal{I}' = (\Omega', \PP', S, T, Y)$, where $\Omega' = \{\omega \in \Omega: \sigma \in S, \tau \in T\}$ and $\PP'[\omega] = \pr{\omega \mid \sigma \in S, \tau \in T}$. The next two steps follow by the convexity of $\tilde{G}^*$ and $x^{1/(1 - \log_2 c)}$, respectively.
\end{proof}

The basic outline of the proof of Lemma~\ref{lem:bob_close_bregman} is similar to that of Lemma~\ref{lem:bob_close_quad}. Once again, we partition $[0, 1]$ into $N$ intervals. Analogously to Equation~\ref{eq:two_sums_quad}, and with $S^{(k(\sigma))}$ defined analogously, we find that
\[\EE{D_G(\mu_{\sigma \tau} \parallel \mu_{\tau})} \le \EE{D_G(\mu_{\sigma} \parallel \mu_{S^{(k(\sigma))}})} + \EE{D_G(\mu_{S^{(k(\sigma))} \tau} \parallel \mu_{\tau})}.\]
As before, we wish to upper bound each summand. However, the fact that the Bregman divergence is now arbitrary introduces complications. First, it is no longer the case that we can directly relate the length of an interval to the Bregman divergence between its endpoints. Second, we consider functions $G$ that become infinitely steep near $0$ and $1$ (such as the negative of Shannon entropy), which makes matters more challenging. This means that we need to be more careful when partitioning $[0, 1]$ into $N$ intervals: see Algorithm~\ref{alg:intervals} for our new approach. Additionally, bounding the second summand involves reasoning carefully about the behavior of the function $G$, which is responsible for the introduction of $\tilde{G}^*$ into the lemma statement. We defer the full proof of Lemma~\ref{lem:bob_close_bregman} to Appendix~\ref{appx:bregman_omitted}.

\section{Connections to Markets} \label{sec:markets}

In this work, we established a natural condition on information structures, \emph{rectangle substitutes}, under which any agreement protocol results in accurate beliefs.
As we saw, a particularly natural class of agreement protocols are \emph{expectation-sharing} protocols, where Alice and Bob take turns stating their current expected value, or discretizations thereof.

Expectation-sharing protocols have close connections to financial markets.
In markets, the actions of traders reveal partial information about their believed value for an asset, i.e., their expectation.
Specifically, a trader's decision about whether to buy or sell, and how much, can be viewed as revealing a discretization of this expectation.
In many theoretical models of markets (see e.g.\ \cite{ostrovsky2012information}) traders eventually reach agreement.
The intuition behind this phenomenon is that a trader who disagrees with the price leaves money on the table by refusing to trade.
Our work thus provides a lens into a well-studied question:\footnote{This is related to the \emph{efficient market hypothesis}, the problem of when market prices reflect all available information, which traces back at least to \citet{fama1970efficient} and \citet{hayek1945use}. Modern models of financial markets are often based on \citet{kyle1985continuous}; we refer the reader to e.g. \cite{ostrovsky2012information} and references therein for further information.} when are market prices accurate?

An important caveat, however, is that traders behave strategically, and may not disclose their true expected value.
For example, a trader may choose to withhold information until a later point when doing so would be more profitable.
Therefore, to interpret the actions of traders as revealing discretized versions of their expected value, one first has to understand the Bayes-Nash equilibria of the market.
\citet{cw16} studies conditions under which traders are incentivized to reveal all of their information on their first trading opportunity.
They call a market equilibrium \emph{all-rush} if every trader is incentivized to reveal their information immediately.
Their main result, roughly speaking, is that there is an all-rush equilibrium if and only if the information structure satisfies \emph{strong substitutes}---a different strengthening of the weak substitutes condition.
This result is specific to settings in which traders have the option to reveal all of their information on their turn---a setting that would be considered trivial from the standpoint of communication theory.

An exciting question for further study is therefore:
under what information structure conditions and market settings is it a Bayes-Nash equilibrium to follow an agreement protocol leading to accurate beliefs?
In other words, what conditions give not only that agreement implies accuracy, but also that the market incentivizes participants to \emph{follow} the protocol?
Together with \citet{cw16}, our work suggests that certain substitutes-like conditions could suffice.

\printbibliography

\appendix
\section{Details Omitted From Section~\ref{sec:prelims}} \label{appx:prelims_omitted}
Above we stated that the weak substitutes condition is not sufficient for agreement to imply accuracy. To see this, consider the following information structure:
\begin{itemize}
    \item Nature flips a coin. Alice's and Bob's signals are each a pair consisting of the outcome of the coin flip and an additional bit:
    \item If the coin lands heads, Alice and Bob are given highly correlated bits as part of their signals: the bits are $(0, 0)$ and $(1, 1)$ with probability 45\% each and $(1, 0)$ and $(0, 1)$ with probability 5\% each.
    \item If the coin lands tails, Alice and Bob are given highly anticorrelated bits as part of their signals: the bits are $(0, 0)$ and $(1, 1)$ with probability 5\% each and $(1, 0)$ and $(0, 1)$ with probability 45\% each.
\end{itemize}
The value $Y$ is the XOR of Alice's and Bob's bits. It can be verified that this information structure satisfies weak substitutes; the intuition is that a majority of the value comes from knowing the outcome of the coin flip, which can be inferred from Alice's (or Bob's) signal alone.

Alice and Bob 0-agree at the very start. However, they are not 0-accurate, since their expectations are either 10\% or 90\%, and the right answer is either 0 or 1. Therefore, weak substitutes alone is insufficient for agreement to imply accuracy.

\section{Details Omitted From Section~\ref{sec:quadratic}} \label{appx:quad_omitted}
\begin{proof}[Proof of Claim~\ref{claim:n_sigma_tau}]
We claim that in fact we can choose the $x_i$'s so that each $x_i$ is in $\brackets{\frac{i}{N} - \frac{1}{2N}, \frac{i}{N} + \frac{1}{2N}}$. This ensures that each interval has length at most $\frac{2}{N}$.

For $x \in [0, 1]$, let $\rho(x)$ be the probability that $x$ is between $\mu_{\sigma}$ and $\mu_{\tau}$, inclusive. Note that $\pr{k(\sigma) \neq k(\tau)} \le \sum_{i = 1}^{N - 1} \rho(x_i)$.

Observe that if $x$ is selected uniformly from $[0, 1]$, the expected value of $\rho(x)$ is equal to $\abs{\mu_{\sigma} - \mu_{\tau}}$, because both quantities are equal to the probability that $x$ is between $\mu_{\sigma}$ and $\mu_{\tau}$.
Therefore, if $(\sigma, \tau)$ is additionally chosen according to $\PP$, we have
\[\EE[{x \leftarrow [0, 1]}]{\rho(x)} = \EE{\abs{\mu_{\sigma} - \mu_{\tau}}} \le \sqrt{\EE{(\mu_{\sigma} - \mu_{\tau})^2}} = \sqrt{\epsilon}.\]
This means that
\[\sum_{i = 1}^{N - 1} \EE[{x \leftarrow \brackets{\frac{i}{N} - \frac{1}{2N}, \frac{i}{N} + \frac{1}{2N}}}]{\rho(x)} = (N - 1) \EE[{x \leftarrow \brackets{\frac{1}{2N}, 1 - \frac{1}{2N}}}]{\rho(x)} \le \sqrt{\epsilon} N.\]
Thus, if each $x_i$ is selected uniformly at random from $\brackets{\frac{i}{N} - \frac{1}{2N}, \frac{i}{N} + \frac{1}{2N}}$, the expected value of $\pr{k(\sigma) \neq k(\tau)}$ would be at most $\sqrt{\epsilon} N$. In particular this means that there exist choices of the $x_i$'s such that $\pr{k(\sigma) \neq k(\tau)} \le \sqrt{\epsilon} N$.
\end{proof}

\begin{proof}[Proof of Lemma~\ref{lem:bob_close_quad}]
Fix a large positive integer $N$ (we will later find it optimal to set $N = \epsilon^{-1/3}$). Consider a partition of $[0, 1]$ into $N$ intervals $[0, x_1), [x_1, x_2), \dots, [x_{N - 1}, 1]$ satisfying the conditions of Claim~\ref{claim:n_sigma_tau}. Let $S^{(k)} := \{\sigma \in \mathcal{S}: x_{k - 1} \le \mu_{\sigma} < x_k\}$. Additionally, let $k(\sigma)$ and $k(\tau)$ be as defined in Claim~\ref{claim:n_sigma_tau}.

Our goal is to upper bound the expectation of $(\mu_{\sigma \tau} - \mu_{\tau})^2$. In pursuit of this goal, we observe that by the Pythagorean theorem, we have
\[\EE{(\mu_{\sigma \tau} - \mu_{\tau})^2} = \EE{(\mu_{\sigma \tau} - \mu_{S^{(k(\sigma))} \tau})^2} + \EE{(\mu_{S^{(k(\sigma))} \tau} - \mu_{\tau})^2}.\]
We now use the rectangle substitutes assumption: for any $k$, by applying Equation~\ref{eq:rec_sub_quad} to $S = S^{(k)}$ and $T = \mathcal{T}$, we know that
\[\EE{(\mu_{\sigma} - \mu_{S^{(k)}})^2 \mid \sigma \in S^{(k)}} \ge \EE{(\mu_{\sigma \tau} - \mu_{S^{(k)} \tau})^2 \mid \sigma \in S^{(k)}}.\]
Taking the expectation over $k$ (i.e.\ choosing each $k$ with probability equal to $\pr{\sigma \in S^{(k)}}$), we have that
\begin{equation} \label{eq:delta_change}
\EE{(\mu_{\sigma} - \mu_{S^{(k(\sigma))}})^2} \ge \EE{(\mu_{\sigma \tau} - \mu_{S^{(k(\sigma))} \tau})^2}.
\end{equation}
Therefore, we have
\begin{equation}
\label{eq:delta_change_2}
\EE{(\mu_{\sigma \tau} - \mu_{\tau})^2} \le \EE{(\mu_{\sigma} - \mu_{S^{(k(\sigma))}})^2} + \EE{(\mu_{S^{(k(\sigma))} \tau} - \mu_{\tau})^2}.
\end{equation}
We will argue use Claim~\ref{claim:n_sigma_tau} to argue that each of these two summands is small. 
The argument regarding the first summand is straightforward: for any $\sigma$, we have that $x_{k(\sigma)} \le \mu_{\sigma}, \mu_{S^{(k(\sigma))}} < x_{k(\sigma) + 1} \le x_{k(\sigma)} + \frac{2}{N}$, which means that $\EE{(\mu_{\sigma} - \mu_{S^{(k(\sigma))}})^2} \le \frac{4}{N^2}$.

We now upper bound the second summand.%
\footnote{The proof below takes sums over $\hat{\tau} \in \mathcal{T}$ and thus implicitly assumes that $\mathcal{T}$ is finite, but the proof extends to infinite $\mathcal{T}$, with sums over $\tau$ replaced by integrals with respect to the probability measure over $\mathcal{T}$.}
For any $\hat{\tau} \in \mathcal{T}$, let $p(\hat{\tau}) = \pr{\tau = \hat{\tau}}$ and $q(\hat{\tau}) = \pr{\tau = \hat{\tau}, k(\sigma) \neq k(\tau)}$. Then $\sum_{\hat{\tau} \in \mathcal{T}} p(\hat{\tau}) = 1$ and $\sum_{\hat{\tau} \in \mathcal{T}} q(\hat{\tau}) \le \sqrt{\epsilon} N$.
Observe that
\begin{align} \label{eq:two_expectations}
\EE{(\mu_{S^{(k(\sigma))} \tau} - \mu_{\tau})^2} &= \sum_{\hat{\tau}} p(\hat{\tau}) \EE{(\mu_{S^{(k(\sigma))} \hat{\tau}} - \mu_{\hat{\tau}})^2 \mid \tau = \hat{\tau}} \nonumber\\
&= \sum_{\hat{\tau}} (p(\hat{\tau}) - q(\hat{\tau})) \EE{(\mu_{S^{(k(\sigma))} \hat{\tau}} - \mu_{\hat{\tau}})^2 \mid \tau = \hat{\tau}, k(\sigma) = k(\hat{\tau})} \nonumber\\
&\qquad + q(\hat{\tau}) \EE{(\mu_{S^{(k(\sigma))} \hat{\tau}} - \mu_{\hat{\tau}})^2 \mid \tau = \hat{\tau}, k(\sigma) \neq k(\hat{\tau})}.
\end{align}

To handle the first expectation, we note that if $k(\sigma) = k(\tau)$, then $\abs{\mu_{S^{(k(\sigma))} \hat{\tau}} - \mu_{\hat{\tau}}} \le \frac{q(\hat{\tau})}{p(\hat{\tau})}$.
To see this, observe
\[p(\hat{\tau}) \mu_{\hat{\tau}} = (p(\hat{\tau}) - q(\hat{\tau})) \mu_{S^{(k(\sigma))} \hat{\tau}} + q(\hat{\tau}) \mu_{\mathcal{S} \setminus S^{(k(\sigma))} \hat{\tau}}~.\]
Rerranging and taking absolute values, we conclude
\[p(\hat{\tau}) \abs{\mu_{S^{(k(\sigma))} \hat{\tau}} - \mu_{\hat{\tau}}} = q(\hat{\tau}) \abs{\mu_{S^{(k(\sigma))} \hat{\tau}} - \mu_{\mathcal{S} \setminus S^{(k(\sigma))}}} \le q(\hat{\tau}).\]
Therefore, recalling $q(\hat\tau) \leq p(\hat\tau)$, we have
\begin{align*}
&(p(\hat{\tau}) - q(\hat{\tau})) \EE{(\mu_{S^{(k(\sigma))} \hat{\tau}} - \mu_{\hat{\tau}})^2 \mid \tau = \hat{\tau}, k(\sigma) = k(\hat{\tau})} 
\le (p(\hat{\tau}) - q(\hat{\tau})) \parens{\frac{q(\hat{\tau})}{p(\hat{\tau})}}^2
\le \frac{q(\hat{\tau})^2}{p(\hat{\tau})}
\le q(\hat{\tau}).
\end{align*}

On the other hand, we can bound the second expectation in Equation~\ref{eq:two_expectations} by $1$. Therefore we have
\[\EE{(\mu_{S^{(k(\sigma))} \tau} - \mu_{\tau})^2} \le \sum_{\hat{\tau}} (q(\hat{\tau}) + q(\hat{\tau})) = 2 \sum_{\hat{\tau}} q(\hat{\tau}) \le 2 \sqrt{\epsilon} N.\]
by Claim~\ref{claim:n_sigma_tau}.
To conclude, we now know that
\[\EE{(\mu_{\sigma \tau} - \mu_{\tau})^2} \le \frac{4}{N^2} + 2 \sqrt{\epsilon} N.\]
Setting $N = \epsilon^{-1/6}$ makes the right-hand side equal to $6\epsilon^{1/3}$, completing the proof.
\end{proof}

\begin{prop} \label{prop:fast_rect}
Consider the following protocol, parametrized by $\epsilon > 0$. Alice and Bob send their initial expectations to each other, rounding to the nearest multiple of $\epsilon$. This protocol entails communicating $O(\log 1/\epsilon)$ bits. At the end of the protocol, Alice and Bob $2\epsilon^2$-agree and are $\epsilon^2$-accurate (with respect to $G(x) = x^2$).
\end{prop}

\begin{proof}
Let $S$ be the set of possible signals of Alice at the end of the protocol which are consistent with the protocol transcript, and define $T$ likewise for Bob. Recall that we use $\mathcal{S}$ and $\mathcal{T}$ to denote the sets of all of Alice's and Bob's possible signals, respectively. We have
\[\EE{(\mu_{\sigma \tau} - \mu_{S \tau})^2} \le \EE{(\mu_{\sigma} - \mu_{S})^2} \le \epsilon^2,\]
since $\mu_{\sigma}$ and $\mu_{S})$ are guaranteed to be within $\epsilon$ of each other by construction. Thus, Bob is $\epsilon^2$-accurate, and likewise for Alice. By the $\frac{1}{2}$-approximate triangle inequality for $G(x) = x^2$, it follows that Alice and Bob $2\epsilon^2$-agree.
\end{proof}

\section{Details Omitted From Section~\ref{sec:bregman}} \label{appx:bregman_omitted}
\begin{proof}[Proof of Proposition~\ref{prop:pythag_bregman}]
Let $g = G'$. We have
\begin{align*}
&\EE{D_G(A \parallel B)} + \EE{D_G(B \parallel C)} - \EE{D_G(A \parallel C)}\\
&= \EE{G(A) - G(B) - (A - B)g(B) + G(B) - G(C) - (B - C)g(C) - G(A) + G(C) + (A - C)g(C)}\\
&= \EE{(A - B)(g(C) - g(B))} = \EE{\EE{(A - B)(g(C) - g(B)) \mid \mathcal{F}}}\\
&= \EE{(g(C) - g(B)) \EE{A - B \mid \mathcal{F}}} = \EE{(g(C) - g(B))(\EE{A \mid \mathcal{F}} - B)} = 0.
\end{align*}
The third-to-last step follows from the fact that $g(C) - g(B)$ is $\mathcal{F}$-measurable (we are using the ``pulling out known factors" property of conditional expectation). The last step follows from the fact that $\EE{A \mid \mathcal{F}} = B$.
\end{proof}

\begin{prop} \label{prop:bregman_facts}
Let $G$ be a differentiable convex function on the interval $[0, 1]$. For all $0 \le a \le b \le 1$, we have
\begin{enumerate}[label=(\roman*)]
    \item \label{item:min_bregman} $\frac{1}{2}(D_G(a \parallel x) + D_G(b \parallel x)) \ge \JB_G(a, b)$ for every $x \in [0, 1]$.
    \item \label{item:triangle} $\JB_G$ satisfies the reverse triangle inequality: for every $x \in [a, b]$, we have $\JB_G(a, x) + \JB_G(x, b) \le \JB_G(a, b)$.
    \item \label{item:shortening} For all $a \le a' \le b' \le b$, we have $\JB_G(a', b') \le \JB_G(a, b)$.
    \item \label{item:fact1} For a random variable $X$ supported on $[a, b]$, we have
    \[\EE{D_G(X \parallel \EE{X})} = \EE{G(X)} - G(\EE{X}) \le 2 \JB_G(a, b).\]
\end{enumerate}
\end{prop}

\begin{proof}
Fact~\ref{item:min_bregman} follows from Proposition~\ref{prop:bregman_exp}. Regarding Fact~\ref{item:triangle}, without loss of generality assume that $x \le \frac{a + b}{2}$ and that $G(x) = G \parens{\frac{a + b}{2}}$ (uniformly adding a constant to the derivative of $G$ does not change any Jensen-Bregman divergence, hence the second assumption). Then $G \parens{\frac{a + x}{2}} \ge G(x)$, so $\JB_G(a, x) \le \frac{G(a) - G(x)}{2}$. Since $\frac{b + x}{2} \ge \frac{a + b}{2}$, we also have that $G \parens{\frac{b + x}{2}} \ge G(x)$, so $\JB_G(b, x) \le \frac{G(b) - G(x)}{2}$. Thus, we have
\[\JB_G(a, x) + \JB_G(b, x) \le \frac{G(a) + G(b)}{2} - G(x) = \frac{G(a) + G(b)}{2} - G \parens{\frac{a + b}{2}} = \JB_G(a, b).\]

\noindent Fact~\ref{item:shortening} follows from Fact~\ref{item:triangle}: we have
\[\JB_G(a, b) = \JB_G(a, a') + \JB_G(a', b') + \JB_G(b', b) \ge \JB_G(a', b').\]

\noindent Regarding the equality in Fact~\ref{item:fact1}, we have
\begin{align*}
\EE{D_G(X \parallel \EE{X})} &= \EE{G(X) - G(\EE{X}) - (X - \EE{X})G'(\EE{X})}\\
&= \EE{G(X) - G(\EE{X})} = \EE{G(X)} - G(\EE{X}),
\end{align*}
where the first step follows from the fact that $\EE{(X - \EE{X}) G'(\EE{X})} = G'(\EE{X}) \EE{X - \EE{X}}$, and $\EE{X - \EE{X}} = 0$.

Regarding the inequality in Fact~\ref{item:fact1}, without loss of generality assume that $\EE{X} \le \frac{a + b}{2}$. By convexity we have that
\[G \parens{\frac{a + b}{2}} \le \frac{\frac{b - a}{2}}{b - \EE{X}} G(\EE{X}) + \frac{\frac{a + b}{2} - \EE{X}}{b - \EE{X}} G(b),\]
so
\begin{align*}
\JB_G(a, b) &= G(a) + G(b) - 2G \parens{\frac{a + b}{2}}\\
&\ge G(a) + G(b) - \frac{b - a}{b - \EE{X}} G(\EE{X}) - \frac{a + b - 2\EE{X}}{b - \EE{X}} G(b)\\
&= G(a) + \frac{\EE{X} - a}{b - \EE{X}} G(b) - \frac{b - a}{b - \EE{X}} G(\EE{X})\\
&= \frac{b - a}{b - \EE{X}} \parens{\frac{b - \EE{X}}{b - a} G(a) + \frac{\EE{X} - a}{b - a} G(b) - G(\EE{X})}\\
&\ge \frac{b - \EE{X}}{b - a} G(a) + \frac{\EE{X} - a}{b - a} G(b) - G(\EE{X}) \ge \EE{G(X)} - G(\EE{X}).
\end{align*}
In the last step we use the fact that for a convex function $f$ and a random variable $X$ defined on an interval $[a, b]$ with mean $\mu$, the maximum possible value of $\EE{f(X)}$ is attained if $X$ is either $a$ or $b$ with the appropriate probabilities.
\end{proof}

\begin{proof}[Proof of Theorem~\ref{thm:agree_bregman}]
Suppose that Alice and Bob do not $\epsilon$-agree at time step $t$, and without loss of generality assume that the next turn (number $t + 1$) is Alice's. We begin by observing that, by Proposition~\ref{prop:bregman_facts}~\ref{item:min_bregman}, we have
\[\EE{D_G(\mu_{\sigma T_t} \parallel \mu_{S_t T_t}) + D_G(\mu_{S_t \tau} \parallel \mu_{S_t T_t})} \ge 2 \EE{\JB_G(\mu_{\sigma T_t}, \mu_{S_t \tau})} > 2\epsilon.\]
Therefore, either $\EE{D_G(\mu_{\sigma T_t} \parallel \mu_{S_t T_t})} \ge \frac{2 \epsilon}{3}$ or $\EE{D_G(\mu_{S_t \tau} \parallel \mu_{S_t T_t})} \ge \frac{4 \epsilon}{3}$.

\paragraph{Case 1:} $\EE{D_G(\mu_{\sigma T_t} \parallel \mu_{S_t T_t})} \ge \frac{2 \epsilon}{3}$. Let us use ``hi," ``lo," and ``md" to denote the events that Alice says ``high," Alice says ``low," and Alice says ``medium," respectively. We have
\begin{align*}
\frac{2 \epsilon}{3} &\le \EE{D_G(\mu_{\sigma T_t} \parallel \mu_{S_t T_t})} = \EE{\EE{D_G(\mu_{\sigma T_t} \parallel \mu_{S_t T_t}) \mid S_t, T_t}}\\
&= \EE{\EE{D_G(\mu_{\sigma T_t} \parallel \mu_{S_t T_t}) \cdot \mathbbm{1}_{\text{hi or lo}} \mid S_t, T_t}} + \EE{\EE{D_G(\mu_{\sigma T_t} \parallel \mu_{S_t T_t}) \cdot \mathbbm{1}_{\text{md}} \mid S_t, T_t}}\\
&\le \EE{\EE{D_G(\mu_{\sigma T_t} \parallel \mu_{S_t T_t}) \cdot \mathbbm{1}_{\text{hi or lo}} \mid S_t, T_t}} + \frac{\epsilon}{2},
\end{align*}
where ``$\mid S_t, T_t$" is short for ``$\mid \sigma \in S_t, \tau \in T_t$," a notation we use throughout the proof. We thus have
\begin{equation}
\label{eq:raf-too-lazy-to-name-this-properly}
\EE{\EE{D_G(\mu_{\sigma T_t} \parallel \mu_{S_t T_t}) \cdot \mathbbm{1}_{\text{hi}} \mid S_t, T_t}} + \EE{\EE{D_G(\mu_{\sigma T_t} \parallel \mu_{S_t T_t}) \cdot \mathbbm{1}_{\text{lo}} \mid S_t, T_t}} \ge \frac{\epsilon}{6}.
\end{equation}

\noindent We now make use of the following lemma.

\begin{lemma} \label{lem:charlie_learns}
Suppose that turn $t + 1$ is Alice's. Let ``hi" denote the event that Alice says ``high." Let $\alpha := \EE{D_G(\mu_{\sigma T_t} \parallel \mu_{S_t T_t}) \cdot \mathbbm{1}_{\text{hi}} \mid S_t, T_t}$. Then
\[\EE{D_G(\mu_{S_{t + 1} T_{t + 1}} \parallel \mu_{S_t T_t}) \cdot \mathbbm{1}_{\text{hi}} \mid S_t, T_t} \ge \frac{\alpha \epsilon}{8M + 2\epsilon}.\]
The analogous statement is true if Alice says ``low," and likewise if it is instead Bob's turn.
\end{lemma}

We assume Lemma~\ref{lem:charlie_learns} and return to prove it afterward.
This lemma translates Equation~\ref{eq:raf-too-lazy-to-name-this-properly} into a statement about how much Charlie learns.
Specifically, we have that
\begin{align*}
&\EE{D_G(\mu_{S_{t + 1} T_{t + 1}} \parallel \mu_{S_t T_t})} = \EE{\EE{D_G(\mu_{S_{t + 1} T_{t + 1}} \parallel \mu_{S_t T_t}) \mid S_t, T_t}}\\
&\ge \EE{\EE{D_G(\mu_{S_{t + 1} T_{t + 1}} \parallel \mu_{S_t T_t}) \cdot \mathbbm{1}_{\text{hi}} \mid S_t, T_t}} + \EE{\EE{D_G(\mu_{S_{t + 1} T_{t + 1}} \parallel \mu_{S_t T_t}) \cdot \mathbbm{1}_{\text{lo}} \mid S_t, T_t}}\\
&\ge \frac{\epsilon}{8M + 2\epsilon} (\EE{\EE{D_G(\mu_{\sigma T_t} \parallel \mu_{S_t T_t}) \cdot \mathbbm{1}_{\text{hi}} \mid S_t, T_t}} + \EE{\EE{D_G(\mu_{\sigma T_t} \parallel \mu_{S_t T_t}) \cdot \mathbbm{1}_{\text{lo}} \mid S_t, T_t}})\\
&\ge \frac{\epsilon^2}{6(8M + 2\epsilon)}.
\end{align*}

\paragraph{Case 2:} $\EE{D_G(\mu_{S_t \tau} \parallel \mu_{S_t T_t})} \ge \frac{4\epsilon}{3}$. Using the Pythagorean theorem to write the same Bregman divergence in two ways, we have that
\begin{align*}
&\EE{D_G(\mu_{S_{t + 1} \tau} \parallel \mu_{S_{t + 1} T_{t + 1}})} + \EE{D_G(\mu_{S_{t + 1} T_{t + 1}} \parallel \mu_{S_t T_t})} = \EE{D_G(\mu_{S_{t + 1} \tau} \parallel \mu_{S_t T_t})}\\
&= \EE{D_G(\mu_{S_{t + 1} \tau} \parallel \mu_{S_t \tau})} + \EE{D_G(\mu_{S_t \tau} \parallel \mu_{S_t T_t})} \ge \EE{D_G(\mu_{S_t \tau} \parallel \mu_{S_t T_t})} \ge \frac{4\epsilon}{3}.
\end{align*}
This means that one of the two summands on the left-hand side is at least $\frac{2\epsilon}{3}$.\\

\textbf{Case 2a:} $\EE{D_G(\mu_{S_{t + 1} \tau} \parallel \mu_{S_{t + 1} T_{t + 1}})} \ge \frac{2\epsilon}{3}$. In that case we have that
\[\EE{D_G(\mu_{S_{t + 2} T_{t + 2}} \parallel \mu_{S_{t + 1} T_{t + 1}})} \ge \frac{\epsilon^2}{6(8M + 2\epsilon)}\]
by the same logic as in Case 1.\\

\textbf{Case 2b:} $\EE{D_G(\mu_{S_{t + 1} T_{t + 1}} \parallel \mu_{S_t T_t})} \ge \frac{2\epsilon}{3} \ge \frac{\epsilon^2}{12\epsilon} \ge \frac{\epsilon^2}{6(8M + 2\epsilon)}$.\\

\noindent In each of our cases, we have that
\begin{align*}
&\EE{D_G(Y \parallel \mu_{S_{t} T_{t}}) - D_G(Y \parallel \mu_{S_{t + 2} T_{t + 2}})} = \EE{D_G(\mu_{S_{t + 2} T_{t + 2}} \parallel \mu_{S_t T_t})}\\
&= \EE{D_G(\mu_{S_{t + 2} T_{t + 2}} \parallel \mu_{S_{t + 1} T_{t + 1}})} + \EE{D_G(\mu_{S_{t + 1} T_{t + 1}} \parallel \mu_{S_t T_t})} \ge \frac{\epsilon^2}{6(8M + 2\epsilon)}.
\end{align*}

\noindent Therefore, the total number of steps until agreement is first reached cannot be more than
\[2 \cdot \frac{M}{\frac{\epsilon^2}{6(8M + 2\epsilon)}} = \frac{24M(4M + \epsilon)}{\epsilon^2}.\]
This completes the proof.
\end{proof}

We now prove Lemma~\ref{lem:charlie_learns}.

\begin{proof}[Proof of Lemma~\ref{lem:charlie_learns}]
We will restrict our probability space to outcomes where Charlie knows $S_t, T_t$ at time $t$ (and thus omit ``$\mid S_t, T_t$" from here on). For convenience, we will let $A := \mu_{\sigma T_t}$ be Alice's expectation (a random variable) and $c := \mu_{S_t T_t}$ be Charlie's expectation (which is a particular number in $[0, 1]$). We will let $\epsilon' := \frac{\epsilon}{2}$, so that if Alice says ``high" then Charlie knows that $A > c$ and that $D_G(A \parallel c) \ge \epsilon'$.

Let $D(x) := D_G(x \parallel c) = G(x) - G(c) - G'(c)(x - c)$, and let $\hat{a}_h := \EE{A \mid \text{hi}}$. Note that if Alice says ``high" then $\mu_{S_{t + 1} T_{t + 1}} = \hat{a}_h$. In our new notation, we may write $\alpha = \EE{D(A) \mid \text{hi}} \cdot \pr{\text{hi}}$, and we wish to show that $D(\hat{a}_h) \cdot \pr{\text{hi}} \ge \frac{\alpha \epsilon'}{2(M + \epsilon')}$. Put otherwise, our goal is to show that
\[\frac{D(\hat{a}_h)}{\EE{D(A) \mid \text{hi}}} \ge \frac{\epsilon'}{2(M + \epsilon')}.\]
For convenience we will let $B$ denote the quantity on the left-hand side.

Let $a_{\text{hmin}}$ be the number larger than $c$ such that $D(a) = \epsilon'$, so that $A \ge a_{\text{hmin}}$ whenever Alice says ``high."\footnote{If $D(a) < \epsilon'$ for all $a > c$ then Alice never says ``high" and the lemma statement is trivial.} Observe that since $D$ is convex (Bregman divergences are convex in their first argument), for a fixed value of $\hat{a}_h$, the value of $\EE{D(A) \mid \text{hi}}$ is maximized when $A$ is either $a_{\text{hmin}}$ or $1$ (with probabilities $\frac{1 - \hat{a}_h}{1 - a_{\text{hmin}}}$ and $\frac{\hat{a}_h - a_{\text{hmin}}}{1 - a_{\text{hmin}}}$, respectively). Therefore we have
\begin{equation} \label{eq:first_b_bound}
B = \frac{D(\hat{a}_h)}{\EE{D(A) \mid \text{hi}}} \ge \frac{D(\hat{a}_h)(1 - a_{\text{hmin}})}{(1 - \hat{a}_h) \epsilon' + (\hat{a}_h - a_{\text{hmin}}) D(1)}.
\end{equation}

\paragraph{Case 1:} $(1 - \hat{a}_h) \epsilon' \ge (\hat{a}_h - a_{\text{hmin}}) D(1)$. In that case we have
\[B \ge \frac{D(\hat{a}_h)(1 - a_{\text{hmin}})}{2(1 - \hat{a}_h) \epsilon'} \ge \frac{\epsilon'(1 - a_{\text{hmin}})}{2(1 - \hat{a}_h) \epsilon'} \ge \frac{1}{2} \ge \frac{\epsilon'}{2(M + \epsilon')}.\]

\paragraph{Case 2:} $(1 - \hat{a}_h) \epsilon' \le (\hat{a}_h - a_{\text{hmin}}) D(1)$. In that case we have

\begin{equation} \label{eq:b_case2}
B \ge \frac{D(\hat{a}_h)(1 - a_{\text{hmin}})}{2(\hat{a}_h - a_{\text{hmin}}) D(1)}.
\end{equation}

\textbf{Case 2a:} $D(1) \le \frac{1 - c}{\hat{a}_h - c}(M + \epsilon')$. Then we have
\[B \ge \frac{D(\hat{a}_h)(1 - a_{\text{hmin}})}{2(\hat{a}_h - a_{\text{hmin}}) \cdot \frac{1 - c}{\hat{a}_h - c}(M + \epsilon')} \ge \frac{\epsilon'}{2(M + \epsilon')} \cdot \frac{(1 - a_{\text{hmin}})(\hat{a}_h - c)}{(\hat{a}_h - a_{\text{hmin}})(1 - c)}.\]
(In the last step we again use that $D(\hat{a}_h) \ge \epsilon$.) Now, it is easy to verify that the second fraction is at least $1$ (this comes down to the fact that $a_{\text{hmin}} \ge c$), so we indeed have that $B \ge \frac{\epsilon'}{2(M + \epsilon')}$.\\

\textbf{Case 2b}: $D(1) \ge \frac{1 - c}{\hat{a}_h - c}(M + \epsilon')$. We claim that for all $x \ge c$, we have that
\begin{equation} \label{eq:m}
D(x) \ge \frac{x - c}{1 - c} D(1) - M.
\end{equation}
To see this, suppose for contradiction that for some $x$ we have $D(x) < \frac{x - c}{1 - c} D(1) - M$. Then
\begin{align*}
G(x) - G(c) - G'(c)(x - c) &< \frac{x - c}{1 - c}(G(1) - G(c) - G'(c)(1 - c)) - M\\
(1 - c)G(x) - (1 - c)G(c) &< (x - c)G(1) - (x - c) G(c) - (1 - c)M\\
G(x) + M &< \frac{(1 - x)G(c) + (x - c)G(1)}{1 - c}.
\end{align*}
On the other hand, we have that both $G(c)$ and $G(1)$ are less than or equal to $G(x) + M$, by definition of $M$. This means that
\[G(1), G(c) < \frac{(1 - x)G(c) + (x - c)G(1)}{1 - c}\]
but this implies that $G(1) < G(c)$ and that $G(c) < G(1)$, a contradiction.\\

Plugging in $x = \hat{a}_h$ into Equation~\ref{eq:m}, we find that
\[D(\hat{a}_h) \ge \frac{\hat{a}_h - c}{1 - c} D(1) - M.\]
Plugging this bound into Equation~\ref{eq:b_case2}, we get that
\begin{align*}
B &\ge \frac{\parens{\frac{\hat{a}_h - c}{1 - c} D(1) - M}(1 - a_{\text{hmin}})}{2(\hat{a}_h - a_{\text{hmin}}) D(1)} = \frac{1 - a_{\text{hmin}}}{2(\hat{a}_h - a_{\text{hmin}})} \cdot \frac{\hat{a}_h - c}{1 - c} \parens{1 - \frac{M}{\frac{\hat{a}_h - c}{1 - c}D(1)}}\\
&\ge \frac{1 - a_{\text{hmin}}}{2(\hat{a}_h - a_{\text{hmin}})} \cdot \frac{\hat{a}_h - c}{1 - c} \parens{1 - \frac{M}{M + \epsilon'}} \ge \frac{\epsilon'}{2(M + \epsilon')},
\end{align*}
where in the second-to-last step we use that $D(1) \ge \frac{1 - c}{\hat{a}_h - c}(M + \epsilon')$ and in the last step we again use the fact that $\frac{(1 - a_{\text{hmin}})(\hat{a}_h - c)}{(\hat{a}_h - a_{\text{hmin}})(1 - c)} \ge 1$.
\end{proof}

\begin{proof}[Proof of Lemma~\ref{lem:bob_close_bregman}]\phantom{}\\
We will partition $[0, 1]$ into a number $N$ of small intervals $I_1 = [x_0 = 0, x_1)$, $I_2 = [x_1, x_2)$, $I_3 = [x_2, x_3)$, \dots, $I_N = [x_{N - 1}, x_N = 1]$ with certain desirable properties (which we will describe below). For $k \in [N]$, we will let $S^{(k)} := \{\sigma \in \mathcal{S}: \mu_{\sigma} \in I_k\}$. For a given $\sigma \in \mathcal{S}$, we will let $k(\sigma)$ be the $k$ such that $\sigma \in S^{(k)}$.

Our goal is to upper bound the expectation of $D_G(\mu_{\sigma \tau} \parallel \mu_{\tau})$. In pursuit of this goal, we observe that by Proposition~\ref{prop:pythag_bregman} we have
\begin{equation} \label{eq:sum_of_parts}
\EE{D_G(\mu_{\sigma \tau} \parallel \mu_{\tau})} = \EE{D_G(\mu_{\sigma \tau} \parallel \mu_{S^{(k(\sigma))} \tau})} + \EE{D_G(\mu_{S^{(k(\sigma))} \tau} \parallel \mu_{\tau})}.
\end{equation}
Now, for any $k$, by applying Equation~\ref{eq:rec_sub_bregman} to $S = S^{(k)}$ and $T = \mathcal{T}$, we know that
\[\EE{D_G(\mu_{\sigma} \parallel \mu_{S^{(k)}}) \mid S^{(k)}} \ge \EE{D_G(\mu_{\sigma \tau} \parallel \mu_{S^{(k)} \tau}) \mid S^{(k)}}.\]
(Here, ``$\mid S^{(k)}$" is short for ``$\mid \sigma \in S^{(k)}$.") This is our only use of the rectangle substitutes assumption. Now, taking the expectation over $k$ (i.e.\ choosing each $k$ with probability equal to $\pr{\sigma \in S^{(k)}}$), we have that
\[\EE{D_G(\mu_{\sigma} \parallel \mu_{S^{(k(\sigma))}})} \ge \EE{D_G(\mu_{\sigma \tau} \parallel \mu_{S^{(k(\sigma))} \tau})}.\]
Together with Equation~\ref{eq:sum_of_parts}, this tells us that
\begin{equation} \label{eq:sum_of_parts_main}
\EE{D_G(\mu_{\sigma \tau} \parallel \mu_{\tau})} \le \EE{D_G(\mu_{\sigma} \parallel \mu_{S^{(k(\sigma))}})} + \EE{D_G(\mu_{S^{(k(\sigma))} \tau} \parallel \mu_{\tau})}.
\end{equation}
Our goal will be to bound the two summands in Equation~\ref{eq:sum_of_parts_main}. We will specify the boundaries of the intervals $I_1, \dots, I_N$ with this goal in mind.\\

On an intuitive level, we are hoping for two things to be true:
\begin{itemize}
\item In order for the first summand to be small, we want $\mu_{\sigma}$ and $\mu_{S^{(k(\sigma))}}$ to be similar in value. In other words, we want each interval is ``short" (for a notion of shortness with respect to $G$ that we are about to discuss).
\item In order for the second summand to be small, we want $\mu_{S^{(k(\sigma))} \tau}$ and $\mu_{\tau}$ to be similar in value. In other words, the estimate of a third party who knows $\tau$ shouldn't change much upon learning $k(\sigma)$. One way to ensure this is by creating the intervals in a way that makes the third party very confident about the value of $k(\sigma)$ before learning it. Intuitively this should be true because Alice and Bob approximately agree, so Alice's estimate is likely to be close to Bob's. However, we must be careful to strategically choose the boundaries of our intervals $x_1, \dots, x_{N - 1}$ so that Alice's and Bob's estimates are unlikely to be on opposite sides of a boundary.\footnote{This limits how many intervals we can reasonably use, which is why we cannot make our intervals arbitrarily short to satisfy the first of our two criteria.}
\end{itemize}

What, formally, do we need for the first summand to be small? For any $k$, we have $\mu_{S^{(k(\sigma))}} = \EE{\mu_{\sigma} \mid \sigma \in S^{(k)}}$. We can apply Proposition~\ref{prop:bregman_facts}~\ref{item:fact1} to the random variable $X = \mu_{\sigma}$ on the probability subspace given by $\sigma \in S^{(k)}$. Since $X$ takes on values in $I_k$, we have that
\begin{equation} \label{eq:2jb}
\EE{D_G(\mu_{\sigma} \parallel \mu_{S^{(k)}}) \mid S^{(k)}} \le 2\JB_G(I_k),
\end{equation}
where $\JB_G(I_k)$ is shorthand for the Jensen-Bregman divergence between the endpoints of $I_k$. Therefore, if $\JB_G(I_k)$ is small for all $k$, then the first summand (which is an expected value of $\EE{D_G(\mu_{\sigma} \parallel \mu_{S^{(k)}}) \mid S^{(k)}}$ over $k \in [N]$) is also small.\\

What about the second summand? As per the intuition above, we wish to choose our boundary points $x_1, \dots, x_{N - 1}$ so that Alice's and Bob's estimates are unlikely to be on opposite sides of any boundary. Let $\mu_- = \min(\mu_{\sigma}, \mu_{\tau})$ be the smaller of the two estimates and $\mu_+ = \max(\mu_{\sigma}, \mu_{\tau})$ be the larger one. We say that $\mu_-, \mu_+$ \emph{thwart} a point $x \in (0, 1)$ if $\mu_- \le x \le \mu_+$ and $\mu_- \neq \mu_+$. We define the \emph{thwart density} of $x$ to be
\[\rho(x) := \pr{\mu_-, \mu_+ \text{ thwart } x}.\]
Roughly speaking, we will choose $x_1, \dots, x_{N - 1}$ such that $\rho(x_k)$ is small on average.\\

We will approach this problem by first creating intervals to satisfy the first criterion (short intervals), without regard to the second, and then modifying them to satisfy the second without compromising the first. Formally, we choose our intervals according to the following algorithm.

\begin{algorithm}[Partitioning {$[0, 1]$} into intervals $I_1, \dots, I_N$] \label{alg:intervals} \phantom{}
\begin{enumerate}[label=(\arabic*)]
    \item Choose points $0 < x_1' < x_2' < \dots < x_{N - 2}' < 1$ such that the $N - 1$ intervals thus created all have Jensen-Bregman divergence between $\beta$ and $\frac{2\beta}{c}$, inclusive, where $\beta$ and $c$ are as in the statement of Lemma~\ref{lem:bob_close_bregman}. ($N$ is not pre-determined; it is defined as one more than the number of intervals created.) (See footnote for why this is possible.\footnote{Define $x_1'$ so that $\JB_G(0, x_1') = \frac{2\beta}{c}$ (this is possible because $\JB_G$ is continuous in its arguments). Define $x_2'$ so that $\JB_G(x_1', x_2') = \frac{2\beta}{c}$. Keep going until an endpoint $x_{N - 3}'$ is defined such that adding $x_{N - 2}'$ as before would leave an interval $(x_{N - 2}', 1)$ with Jensen-Bregman divergence less than $\frac{2\beta}{c}$. Now, instead of defining $x_{N - 2}'$ in this way, define it so that $\JB_G(x_{N - 3}', x_{N - 2}') = \JB_G(x_{N - 2}', 1)$. Since $\JB_G(x_{N - 3}', 1) \ge \frac{2\beta}{c}$, the $c$-approximate triangle inequality that we have by assumption tells us that $\JB_G(x_{N - 3}', x_{N - 2}') = \JB_G(x_{N - 2}', 1) \ge \beta$.})
    
    \item Let $x_0' := 0, x_{N - 1}' := 1$ for convenience. Define $I_k' := [x_{k - 1}', x_k']$. For $k \in [N - 1]$, let $\alpha_k := \inf_{x \in I_k'} \rho(x)$. Let $x_k \in I_k'$ be such\footnote{If the infimum is achieved (e.g.\ if the space of signals to Alice and Bob is finite), then we can set $x_k := \arg \min_x \rho(x)$. Our algorithm works in more generality, at the expense of a factor of $2$ in our final bound. Note that by replacing $2$ with a smaller constant can arbitrarily reduce this factor.} that $\rho(x_k) \le 2\alpha_k$.
    
    \item Return the intervals $I_1 = [0, x_1), I_2 = [x_1, x_2), \dots, I_N = [x_{N - 1}, 1]$.
\end{enumerate}
\end{algorithm}

\noindent We begin by observing that for any $k \in [N]$, we have
\[\JB_G(I_k) = \JB_G(x_{k - 1}, x_k) \le \JB_G(x_{k - 2}', x_k') \le \frac{1}{c}(\JB_G(x_{k - 2}', x_{k - 1}') + \JB_G(x_{k - 1}', x_k')) \le \frac{4\beta}{c^2}\]
where for convenience we define $x_{-1}' := 0, x_N' := 1$.
Therefore, by Equation~\ref{eq:2jb}, we have
\begin{equation} \label{eq:bound_1}
\EE{D_G(\mu_{\sigma} \parallel \mu_{S^{(k(\sigma))}})} \le \frac{8\beta}{c^2}.
\end{equation}

It remains to bound the second summand of Equation~\ref{eq:sum_of_parts_main}, $\EE{D_G(\mu_{S^{(k(\sigma))} \tau} \parallel \mu_{\tau})}$, which is the bulk of the proof. We proceed in two steps:
\begin{enumerate}[label=(\arabic*)]
    \item (Lemma~\ref{lem:jbgi_small}) We show that $\sum_{k = 1}^N \alpha_k$ is small. This means that Alice's and Bob's estimates are unlikely to lie on opposite sides of some boundary point $x_k$. As a consequence, Bob is highly likely to know $k(\sigma)$ with a lot of confidence
    \item (Lemma~\ref{lem:second_summand}) We bound the second summand as a function of $\sum_{k = 1}^N \alpha_k$. The intuition is that if $\sum_k \alpha_k$ is small, then Bob is highly likely to know $k(\sigma)$ with a lot of confidence, which means that he does not learn too much from learning $k(\sigma)$.
\end{enumerate}

\noindent We begin with the first step; recall our notation $\mu^- := \min(\mu_{\sigma}, \mu_{\tau})$ and $\mu^+ := \max(\mu_{\sigma}, \mu_{\tau})$.

\begin{lemma} \label{lem:jbgi_small}
\[2 \sum_{k = 1}^N \alpha_k \le 4 \parens{\frac{\epsilon}{\beta c}}^{1/(1 - \log_2 c)}.\]
\end{lemma}

\begin{proof}
We use the following claim, whose proof we provide afterward.

\begin{claim} \label{claim:alpha_2}
Let $I = [x^-, x^+]$ be any sub-interval of $[0, 1]$ and let $\alpha = \inf_{x \in I} \rho(x)$. Then there is an increasing sequence of points $z_0 := x^-, z_1, z_2, \dots, z_{L - 1}, z_L := x^+$, such that for every $\ell \in [L]$, $\pr{\mu_- \le z_{\ell - 1}, \mu_+ \ge z_\ell} \ge \frac{\alpha}{2}$, and where
\[L \le \frac{2}{\alpha} \sum_{\ell \in [L]} \pr{\mu_- \le z_{\ell - 1} < \mu_+ \le z_\ell}.\]
\end{claim}

We apply Claim~\ref{claim:alpha_2} to the intervals $I_1', \dots, I_{N - 1}'$, with $\alpha = \alpha_k$. Let $z_{k, 0}, \dots, z_{k, L_k}$ be the points whose existence the claim proves, and let $r_k := \sum_{\ell \in [L_k]} \pr{\mu_- \le z_{k, \ell - 1} < \mu_+ \le z_{k, \ell}}$, so that $L_k \le \frac{2}{\alpha_k} r_k$. Observe that $\sum_k r_k \le 1$, because the intervals $(z_{k, \ell - 1}, z_{k, \ell}]$ are disjoint for all $k, \ell$. We make the following claim (we provide the proof afterward).

\begin{claim} \label{claim:alphak_bound}
\begin{equation} \label{eq:alphak_bound}
\sum_{k \in [N - 1]} r_k \parens{\frac{\alpha_k}{2r_k}}^{1 - \log_2 c} \le \frac{\epsilon}{\beta c}.
\end{equation}
\end{claim}

\noindent We may rewrite Equation~\ref{eq:alphak_bound} as
\[\parens{\sum_{k \in [N - 1]} r_k \parens{\frac{\alpha_k}{2r_k}}^{1 - \log_2 c}}^{1/(1 - \log_2 c)} \le \parens{\frac{\epsilon}{\beta c}}^{1/(1 - \log_2 c)}.\]

Recall that $\sum_k r_k \le 1$. Scaling the $r_k$'s to add to $1$ decreases the left-hand side above, so we may assume that $\sum_k r_k = 1$. Note that $x^{1 - \log_2 c}$ is convex. Thus, by using a weighted Jensen inequality on the left-hand side with weights $r_k$, we find that
\[\frac{1}{2} \sum_k \alpha_k = \sum_k r_k \cdot \frac{\alpha_k}{2r_k} \le \parens{\sum_{k \in [N - 1]} r_k \parens{\frac{\alpha_k}{2r_k}}^{1 - \log_2 c}}^{1/(1 - \log_2 c)} \le \parens{\frac{\epsilon}{\beta c}}^{1/(1 - \log_2 c)}.\]
This completes the proof of Lemma~\ref{lem:jbgi_small}.
\end{proof}

\begin{proof}[Proof of Claim~\ref{claim:alpha_2}]
Let $z_1 = \inf \{z: \pr{\mu_- \le z_0 < \mu_+ \le z} \ge \frac{\alpha}{2}\}$, or $x^+$ if this number does not exist or is larger than $x^+$.
Note that $\pr{\mu_- \le z_0 < \mu_+} \ge \alpha$, as we have $\rho(z_0) = \pr{\mu_- \le z_0 < \mu_+} + \pr{\mu_- < z_0 = \mu_+} \ge \alpha$, so if the first term were less than $\alpha$ we would have some $z' > z_0$ with $\rho(z') < \alpha$.
On the other hand, $\pr{\mu_- \le z_0 < \mu_+ < z_1} \le \frac{\alpha}{2}$, since
\[\pr{\mu_- \le z_0 < \mu_+ < z_1} = \lim_{z \to z_1 \text{ from below}} \pr{\mu_- \le z_0 < \mu_+ \le z}\]
and if the right-hand side were more than $\frac{\alpha}{2}$ then that would contradict the definition of $z_1$ as an infimum.
Therefore, $\pr{\mu_- \le z_0, \mu_+ \ge z_1} \ge \frac{\alpha}{2}$.

If $z_1 = x^+$, we are done. Otherwise, let $z_2 = \inf \{z: \pr{\mu_- \le z_1 < \mu_+ \le z} \ge \frac{\alpha}{2}\}$. Then $\pr{\mu_- \le z_1, \mu_+ \ge z_2} \ge \frac{\alpha}{2}$. Define $z_3$ analogously, and so forth.

All that remains to show is the upper bound on $L$. This is where we use the fact that (by construction) $\pr{\mu_- \le z_{\ell - 1} < \mu_+ \le z_\ell} \ge \frac{\alpha}{2}$. Summing over all $\ell$, we have
\[\sum_{\ell \in [L]} \pr{\mu_- \le z_{\ell - 1} < \mu_+ \le z_\ell} \ge \frac{\alpha}{2} L,\]
which (after rearranging) completes the proof.
\end{proof}

\begin{proof}[Proof of Claim~\ref{claim:alphak_bound}]
First note that by construction, $\JB_G(I_k') \ge \beta$ for all $k$. By repeated use of the $c$-approximate triangle inequality,\footnote{We sub-divide $I_k'$ into $[z_{k, 0}, z_{k, L_k/2}]$ and $[z_{k, L_k/2}, z_{k, L}]$, then subdivide each of these, and so on.} we find that
\[\sum_{\ell \in [L_k]} \JB_G(z_{k, \ell - 1}, z_{k, \ell}) \ge c^{\ceil{\log_2 L_k}} \JB_G(I_k') \ge c^{1 + \log_2 \frac{2r_k}{\alpha_k}} \JB_G(I_k') \ge c^{1 + \log_2 \frac{2r_k}{\alpha_k}} \beta = c \parens{\frac{2r_k}{\alpha_k}}^{\log_2 c} \beta.\]
On the other hand, we have
\begin{align*}
\epsilon &\ge \EE{\JB_G(\mu_{\sigma}, \mu_{\tau})} = \sum_{\sigma, \tau} \pr{\sigma, \tau} \JB_G(\mu_{\sigma}, \mu_{\tau}) \ge \sum_{\sigma, \tau} \pr{\sigma, \tau} \sum_{\substack{k, \ell: \mu_- \le z_{k, \ell - 1} \\ \mu_+ \ge z_{k, \ell}}} \JB_G(z_{k, \ell - 1}, z_{k, \ell})\\
&= \sum_{k, \ell} \pr{\mu_- \le z_{k, \ell - 1}, \mu_+ \ge z_{k, \ell}} \JB_G(z_{k, \ell - 1}, z_{k, \ell}) \ge \sum_{k, \ell} \frac{\alpha_k}{2} \JB_G(z_{k, \ell - 1}, z_{k, \ell}).
\end{align*}
Here, the third step follows by the reverse triangle inequality (Fact~\ref{item:triangle} of Proposition~\ref{prop:bregman_facts}) and the fourth step follows by rearranging the order of summation.\footnote{The case that the space of signals is infinite is identical except that the summation is replaced by an integral over the probability space.} Combining the last two facts gives us that
\[\epsilon \ge \sum_k \frac{\alpha_k}{2} \cdot c \parens{\frac{2r_k}{\alpha_k}}^{\log_2 c} \beta = \sum_k r_k  \parens{\frac{2r_k}{\alpha_k}}^{\log_2 c - 1} \beta c,\]
which rearranges to the desired identity.
\end{proof}

We are now ready to bound the second summand, i.e.\ $\EE{D_G(\mu_{S^{(k(\sigma))} \tau} \parallel \mu_{\tau})}$, where $k(\sigma)$ is the $k$ such that Alice's estimate $\mu_{\sigma}$ lies in $I_k$. For convenience we will define $k(\tau)$ for Bob by analogy as the $k$ such that $\mu_{\tau}$ lies in $I_k$. By Lemma~\ref{lem:jbgi_small} and the preceding discussion, we know that
\[\pr{k(\sigma) \neq k(\tau)} \le 4 \parens{\frac{\epsilon}{\beta c}}^{1/(1 - \log_2 c)}.\]

\begin{lemma} \label{lem:second_summand}
Let $Q = \pr{k(\sigma) \neq k(\tau)}$. Then
\[\EE{D_G(\mu_{S^{(k(\sigma))} \tau} \parallel \mu_{\tau})} \le 2\tilde{G}^*(Q).\]
\end{lemma}

The key idea is that because $k(\sigma) = k(\tau)$ with probability near $1$, learning $k(\sigma)$ is unlikely to make Bob update his estimate much.

\begin{proof}
Consider any signal $\hat{\tau} \in \mathcal{T}$ and let $p(\hat{\tau}) = \pr{\tau = \hat{\tau}}$. We have\footnote{This proof takes sums over $\hat{\tau} \in \mathcal{T}$ and thus implicitly assumes that $\mathcal{T}$ is finite, but the proof extends to infinite $\mathcal{T}$, with sums over $\tau$ replaced by integrals with respect to the probability measure over $\mathcal{T}$.}
\[\EE{D_G(\mu_{S^{(k(\sigma))} \tau} \parallel \mu_{\tau})} = \sum_{\hat{\tau} \in \mathcal{T}} p(\hat{\tau}) \EE{D_G(\mu_{S^{(k(\sigma))} \hat{\tau}} \parallel \mu_{\hat{\tau}}) \mid \tau = \hat{\tau}}.\]
Note that $\mu_{\hat{\tau}} = \EE{\mu_{S^{(k(\sigma))} \hat{\tau}} \mid \tau = \hat{\tau}}$, so by Proposition~\ref{prop:bregman_facts} we have that
\[\EE{D_G(\mu_{S^{(k(\sigma))} \tau} \parallel \mu_{\tau})} = \sum_{\hat{\tau} \in \mathcal{T}} p(\hat{\tau}) \parens{\EE{G(\mu_{S^{(k(\sigma))}\hat{\tau}}) \mid \tau = \hat{\tau}} - G(\mu_{\hat{\tau}})}.\]
Let $q(\hat{\tau}) = \pr{\tau = \hat{\tau}, k(\sigma) \neq k(\hat{\tau})}$, so $\sum_{\hat{\tau} \in \mathcal{T}} q(\hat{\tau}) = Q$. Then
\begin{align*}
&\EE{G(\mu_{S^{(k(\sigma))}\hat{\tau}}) \mid \tau = \hat{\tau}} - G(\mu_{\hat{\tau}}) =\\
&\frac{p(\hat{\tau}) - q(\hat{\tau})}{p(\hat{\tau})} \parens{\EE{G(\mu_{S^{(k(\hat{\tau}))} \hat{\tau}}) - G(\mu_{\hat{\tau}})}} + \frac{q(\hat{\tau})}{p(\hat{\tau})} \parens{\EE{G(\mu_{S^{(k(\sigma))}\hat{\tau}}) \mid \tau = \hat{\tau}, k(\sigma) \neq k(\hat{\tau})} - G(\mu_{\hat{\tau}})}.
\end{align*}
The second term is at most $\frac{q(\hat{\tau})}{p(\hat{\tau})} M$, since $M$ is the range of $G$. To bound the first term, we note that $\mu_{S^{(k(\hat{\tau}))}\hat{\tau}}$ cannot differ from $\mu_{\hat{\tau}}$ by more than $\frac{q(\hat{\tau})}{p(\hat{\tau}) - q(\hat{\tau})}$, as otherwise the average value of $\mu_{S^{(k(\sigma))} \hat{\tau}}$ could not be $\mu_{\hat{\tau}}$. Therefore, $\EE{G(\mu_{S^{(k(\hat{\tau}))} \hat{\tau}}) - G(\mu_{\hat{\tau}})}$ is bounded by the largest possible difference in $G$-values of two points that differ by at most $\frac{q(\hat{\tau})}{p(\hat{\tau}) - q(\hat{\tau})}$. Therefore, we have
\begin{align*}
\EE{D_G(\mu_{S^{(k(\sigma))} \tau} \parallel \mu_{\tau})} &\le \sum_{\hat{\tau} \in \mathcal{T}} p(\hat{\tau}) \parens{\frac{p(\hat{\tau}) - q(\hat{\tau})}{p(\hat{\tau})} \tilde{G} \parens{\frac{q(\hat{\tau})}{p(\hat{\tau}) - q(\hat{\tau})}} + \frac{q(\hat{\tau})}{p(\hat{\tau})} M}\\
&\le QM + \sum_{\hat{\tau} \in \mathcal{T}} (p(\hat{\tau}) - q(\hat{\tau})) \tilde{G} \parens{\frac{q(\hat{\tau})}{p(\hat{\tau}) - q(\hat{\tau})}},
\end{align*}
where $\tilde{G}$ is defined as in the statement of Lemma~\ref{lem:second_summand}. If $G$ is symmetric on $[0, 1]$, then $\tilde{G}(x) = G(0) - G(x)$ for $x \le \frac{1}{2}$ and $M$ otherwise. This is a concave function, but $\tilde{G}$ is not in general concave. However, consider $\tilde{G}^*$ as defined in the lemma statement, so $\tilde{G}(x) \le \tilde{G}^*(x)$ for all $x$. Then
\begin{align*}
\EE{D_G(\mu_{S^{(k(\sigma))} \tau} \parallel \mu_{\tau})} &\le QM + \sum_{\hat{\tau} \in \mathcal{T}} (p(\hat{\tau}) - q(\hat{\tau})) \tilde{G}^* \parens{\frac{q(\hat{\tau})}{p(\hat{\tau}) - q(\hat{\tau})}}\\
&\le QM + \parens{\sum_{\hat{\tau} \in \mathcal{T}} (p(\hat{\tau}) - q(\hat{\tau}))} \cdot \tilde{G}^* \parens{\frac{\sum_{\hat{\tau} \in \mathcal{T}} q(\hat{\tau})}{\sum_{\hat{\tau} \in \mathcal{T}} (p(\hat{\tau}) - q(\hat{\tau}))}}\\
&= QM + (1 - Q) \tilde{G}^* \parens{\frac{Q}{1 - Q}} \le Q M + \tilde{G}^*(Q) \le 2\tilde{G}^*(Q).
\end{align*}
Here, the second step follows by Jensen's inequality with terms $\frac{q(\hat{\tau})}{p(\hat{\tau}) - q(\hat{\tau})}$ and weights $p(\hat{\tau}) - q(\hat{\tau})$, the second-to-last step follows from the fact that $\tilde{G}^*$ is convex and $\tilde{G}^*(0) = 0$, and the last step follows from the fact that $\tilde{G}^*$ is convex and $\tilde{G}^*(1) = M$.
\end{proof}

Since $Q \le 4 \parens{\frac{\epsilon}{\beta c}}^{1/(1 - \log_2 c)}$, combining Lemma~\ref{lem:second_summand} with Equation~\ref{eq:bound_1} gives us the following result.
\[\EE{D_G(\mu_{\sigma \tau} \parallel \mu_{\tau})} \le \frac{8\beta}{c^2} + 2 \tilde{G}^* \parens{4 \parens{\frac{\epsilon}{\beta c}}^{1/(1 - \log_2 c)}}.\]
Noting that $\tilde{G}^*$ is concave and $c^{-1/(1 - \log_2 c)} \le 2$ (which is true for all $0 < c < 1$) completes the proof of Lemma~\ref{lem:bob_close_bregman}.
\end{proof}

\section{Alternative Definitions of Agreement and Accuracy} \label{appx:alternative_defs}
For arbitrary Bregman divergences, there are several notions of agreement and accuracy that are worth considering. Before we discuss these, we make a note about the order of arguments in a Bregman divergence. In our context, it makes the most sense to talk of the Bregman divergence \emph{from a more informed estimate to a less informed estimate}. By a ``more informed estimate" we mean a finer-grained one, i.e.\ one that is informed by more knowledge. For example, in terms of estimating $Y$ in the context of this work explores, $Y$ is more informed than $\mu_{\sigma \tau}$, which is more informed than $\mu_{\sigma}$ and $\mu_{S \tau}$, which are each more informed than $\mu_{ST}$, which is more informed than $\EE{Y}$.

To see that this is the natural order of the arguments, recall that Bregman divergences are motivated by the property that they elicit the mean (see Proposition~\ref{prop:bregman_exp}): if an agent who gives an estimate of $x$ for the value of a random variable $Y$ incurs a loss of $D_G(Y \parallel x)$, then the agent minimizes their expected loss by reporting $x = \EE{Y}$. This means that the expert ought to report the expected value of $Y$ given the information that the expert knows.

This means that given two estimates of $Y$, $Z_1$ and $Z_2$, of which $Z_1$ is more informed, the quantity $D_G(Z_1 \parallel Z_2)$ has a natural interpretation: it is the expected amount the expert gains by learning more and refining their estimate from $Z_2$ to $Z_1$. This follows by the Pythagorean theorem:
\[\EE{D_G(Z_1 \parallel Z_2)} = \EE{D_G(Y \parallel Z_2)} - \EE{D_G(Y \parallel Z_1)}.\]

\subsection{Alternative Definitions of Agreement}
One important motivation for using the Jensen-Bregman divergence to the midpoint as the definition of agreement is that this quantity serves as a lower bound on the expected amount that Charlie disagrees with Alice and Bob. Formally:

\begin{defin} \label{def:agree_charlie}
Let $a$, $b$, and $c$ be Alice's, Bob's, and Charlie's expectations, respectively (these are random variables on $\Omega$). Alice and Bob \emph{$\epsilon$-agree with Charlie} if $\frac{1}{2}(\EE{D_G(a \parallel c) + D_G(b \parallel c)}) \le \epsilon$.
\end{defin}

(This is the order of arguments because Alice and Bob are more informed than Charlie.) By Proposition~\ref{prop:bregman_facts}~\ref{item:min_bregman}, we know that \textbf{if Alice and Bob $\epsilon$-agree with Charlie then they $\epsilon$-agree}.

As it happens the fact that under this (stronger) definition of agreement implies accuracy under rectangle substitutes follows immediately:

\begin{prop}
Let $\mathcal{I} = (\Omega, \PP, \mathcal{S}, \mathcal{T}, Y)$ be an information structure that satisfies rectangle substitutes. For any communication protocol that causes Alice and Bob to $\epsilon$-agree with Charlie on $\mathcal{I}$, Alice and Bob are $2 \epsilon$-accurate after the protocol terminates.
\end{prop}

\begin{proof}
Let $S$ be the set of possible signals of Alice at the end of the protocol which are consistent with the protocol transcript, and define $T$ likewise for Bob. Recall that Charlie's expectation is $\mu_{ST}$. We have
\[\EE{D_G(\mu_{\sigma \tau} \parallel \mu_{S \tau})} \le \EE{D_G(\mu_{\sigma T} \parallel \mu_{ST})} \le \EE{D_G(\mu_{\sigma T} \parallel \mu_{ST})} + \EE{D_G(\mu_{S \tau} \parallel \mu_{ST})} \le 2 \epsilon,\]
where the first inequality follows by rectangle substitutes and the last inequality follows because Alice and Bob $\epsilon$-agree with Charlie.
\end{proof}

The drawback of Definition~\ref{def:agree_charlie} is that it is not so much a definition of Alice and Bob's agreement with each other, so much as a definition of agreement with respect to the protocol being run (since Charlie only exists within the context of the protocol). Put otherwise, it is impossible to determine whether Alice and Bob $\epsilon$-agree with Charlie simply by knowing Alice and Bob's expectations; one must also know Charlie's expectation, which cannot be determined from Alice's and Bob's expectations. The question ``how far from agreement are Alice and Bob if Alice believes 25\% and Bob believes 30\%?" makes sense in the context of $\epsilon$-agreement, but not in the context of $\epsilon$-agreement with Charlie.\\

A different notion of agreement, which (like $\epsilon$-agreement) only depends on Alice's and Bob's expectations, uses the \emph{symmetrized Bregman divergence} between these expectations: $\frac{1}{2}(D_G(a \parallel b) + D_G(b \parallel a))$.

\begin{defin}
Let $a$ and $b$ be Alice's and Bob's expectations, respectively (these are random variables on $\Omega$). Alice and Bob satisfy \emph{symmetrized $\epsilon$-agreement} if $\frac{1}{2}(D_G(a \parallel b) + D_G(b \parallel a))$.
\end{defin}

By Proposition~\ref{prop:bregman_facts}~\ref{item:shortening}, we know that \textbf{if Alice and Bob satisfy symmetrized $\epsilon$-agreement then they $\epsilon$-agree}.

In our context, symmetrized Bregman divergence is less natural than Jensen-Bregman divergence. This is symmetrized Bregman divergence (unlike Jensen-Bregman divergence) does not seem to closely relate to our previous discussion of the Bregman divergence from a more informed to a less informed estimate being most natural.

\subsection{Alternative Notions of Accuracy}
Our definition of Alice's accuracy as the expected Bregman divergence from the truth $\mu_{\sigma \tau}$ to Alice's expectation seems like the most natural one. However, one may desire a definition of accuracy that takes both Alice's and Bob's expectations into account, judging the pair's accuracy based on their consensus belief, rather than each of their individual beliefs. For instance, one could say that Alice and Bob are \emph{$\epsilon$-midpoint-accurate} if $\EE{D_G \parens{\mu_{\sigma \tau} \parallel \frac{a + b}{2}}} \le \epsilon$. By this definition, Alice's and Bob's expectations could individually be far from the truth, but they are considered accurate because the average of their expectations is close to correct.

\begin{prop} \label{prop:midpoint_accurate}
If Alice and Bob are $\epsilon$-accurate, then they are $2 \epsilon$-midpoint-accurate.
\end{prop}

\begin{proof}
Observe that for all $a, b, y$ it is the case that
\[D_G \parens{y \parallel \frac{a + b}{2}} \le \max(D_G(y \parallel a), D_G(y \parallel b) \le D_G(y \parallel a) + D_G(y \parallel b).\]
The first inequality is true simply because $\frac{a + b}{2}$ lies in between $a$ and $b$. Therefore,
\[\EE{D_G \parens{y \parallel \frac{a + b}{2}}} \le \EE{D_G(y \parallel a) + D_G(y \parallel b)} \le 2\epsilon.\]
\end{proof}

Another natural choice for Alice's and Bob's consensus belief is the QA pool (see \cite{nr21}). Proposition~\ref{prop:midpoint_accurate} likewise holds for the QA pool in place of the midpoint, and indeed holds for any choice of consensus belief that is guaranteed to lie in between Alice's and Bob's expectations. Thus, any such definition will be weaker than our definition of $\epsilon$-accuracy for Alice and Bob (up to a constant factor).\\

To summarize, among the above definitions of agreement, $\epsilon$-agreement is the weakest; and among the above definitions of accuracy, Alice's and Bob's $\epsilon$-accuracy is the strongest. This is an indication of strength for Theorem~\ref{thm:agreement_accurate_bregman}: it starts from a relatively weak premise and reaches a relatively strong conclusion.

\section{Implications for Communication Complexity} \label{appx:comm}
Our results can be framed in a communication complexity context, where they imply that ``substitutable'' functions can be computed with probability $1-\delta$ (over the inputs) with a transcript length depending only on $\delta$.
This is a nonstandard and weak notion of computing the function, but sketching the reduction may inspire future work on connections between substitutes and communication complexity.

In a classic deterministic communication complexity setup (e.g.\ \cite{rao2020communication}), Alice holds $\sigma \in \mathcal{S}$, Bob holds $\tau \in \mathcal{T}$, and the goal is to compute some function $g: \mathcal{S} \times \mathcal{T} \to \{0,1\}$ using a communication protocol (see Section \ref{subsec:agreement-protocols}).
Our setting captures this model when $Y = g(\sigma,\tau)$.
Observe that in this case, $Y = \mu_{\sigma \tau}$, i.e.\ Alice and Bob's information together determine $Y$ completely.
A communication protocol defines its output by a function $h: \Pi \to \{0,1\}$ where $\Pi$ is the space of transcripts.
We can simply let $h(\pi) = \text{round}(\mu_{ST})$, i.e.\ rounding the \emph{ex post} expectation $\EE{Y \mid \pi} = \mu_{ST}$ to either zero or one.
This is equivalent to the belief of ``Charlie'', or the common knowledge of Alice and Bob after the protocol is completed.

\begin{defin}[Rectangle substitutes, $(1-\delta)$-computes]
  Given a function $g$ and a distribution $\mathcal{D}$ over $\mathcal{S} \times \mathcal{T}$, we say $(g,\mathcal{D})$ satisfy \emph{rectangle substitutes} if the corresponding information structure with $Y = g(\sigma,\tau)$ satisfies rectangle substitutes (Definition \ref{def:rect_subs_quad}).
  We say a protocol $(1-\delta)$-computes $g$ over $\mathcal{D}$ if, with probability at least $1-\delta$ over $(\sigma,\tau) \sim \mathcal{D}$, the protocol has $h(\pi) = g(\sigma,\tau)$.
\end{defin}

By our results, under rectangle substitutes $(g,\mathcal{D})$, any agreement protocol approximately computes $g$ over $\mathcal{D}$.
More precisely, using a fast substitutes-agreement protocol similar to Proposition \ref{prop:fast_rect}, we obtain the following.
\begin{corollary}
  Suppose $(g, \mathcal{D})$ satisfy rectangle substitutes.
  Then for every $\delta \in (0,1)$, there is a deterministic communication protocol using $O(\log(1/\delta))$ bits of communication that $(1-\delta)$-computes $g$ over $\mathcal{D}$.
\end{corollary}
\begin{proof}
  In round one, Alice sends her current expectation $\mu_{\sigma}$ rounded to a multiple of $\epsilon$; call this message $A$.
  In round two, Bob sends his updated expectation $\mu_{S \tau}$ rounded to a multiple of $\epsilon$; call this message $B$.
  The protocol then halts, and the output is $B$ rounded to either zero or one.
  It uses $O(\log(1/\epsilon))$ bits.
  Let $S,T$ be the random rectangle associated with the protocol.
  
  By construction, $|\mu_{\sigma} - A| \leq \epsilon$, and $\mu_{S}$ is the expectation of $Y$ conditioned on $A$, so it follows that $|\mu_{\sigma} - \mu_{S}| \leq \epsilon$.
  Using substitutes (just as in Proposition~\ref{prop:fast_rect}),
  \[ \EE{(\mu_{\sigma \tau} - \mu_{S \tau})^2} \leq \EE{(\mu_{\sigma} - \mu_{S})^2} \leq \epsilon^2 . \]
  By construction, $|B - \mu_{S\tau}| \leq \epsilon$.
  Therefore, by the $\tfrac{1}{2}$-approximate triangle inequality for squared distance (e.g.\ Proposition~\ref{prop:triangle})),
  \[ \EE{(\mu_{\sigma \tau} - B)^2}
      ~\leq~ 2\EE{(\mu_{\sigma \tau} - \mu_{S \tau})^2} + 2\EE{(\mu_{S \tau} - B)^2}
      ~\leq~ 2\epsilon^2 . \]
  Now, the protocol is incorrect if $|B - \mu_{\sigma \tau}| \geq \tfrac{1}{2}$.
  Using Markov's inequality,
  \begin{align*}
    \Pr[|B - \mu_{\sigma \tau}| \geq \tfrac{1}{2}]
    &=    \Pr[(B - \mu_{\sigma \tau})^2 \geq \tfrac{1}{4}]  \\
    &\leq 4\EE{(B - \mu_{\sigma \tau})^2}  \\
    &\leq 8\epsilon^2 .
  \end{align*}
  Therefore, given $\delta \in (0,1)$, we run the protocol with $\epsilon = \sqrt{\delta/8}$.
  The probability of an incorrect output is at most $\delta$, and we use $O(\log(1/\epsilon) = O(\log(1/\delta))$ bits of communication.
\end{proof}
\end{document}